\documentclass{article}
\usepackage[utf8]{inputenc}
\usepackage{amsthm}
\usepackage{thmtools}
\usepackage{amssymb}
\usepackage{amsmath}
\usepackage{hyperref}
\usepackage{bbold}
\usepackage{mathtools} 
\usepackage{cancel}
\usepackage[ruled, linesnumbered]{algorithm2e}
\usepackage[margin=1.1in]{geometry}

\usepackage[sorting=none, url=false, isbn=false]{biblatex}
\usepackage{tikz}

\declaretheorem{theorem}
\declaretheorem[sibling=theorem]{lemma, proposition, corollary}
\declaretheorem{assumption}
\DeclareMathOperator{\diff}{d \!}
\DeclareMathOperator\Tr{Tr}
\newcommand\id{\mathbb{1}}

\DeclareMathOperator\Bounded{\mathcal{B}}

\newcommand\R{\mathbb{R}}
\newcommand\C{\mathbb{C}}

\providecommand{\od}[3][]{\ensuremath{
\ifinner
\tfrac{\diff{^{#1}}#2}{\diff{{#3}^{#1}}}
\else
\dfrac{\diff{^{#1}}#2}{\diff{{#3}^{#1}}}
\fi
}}

\providecommand{\dod}[3][]{\ensuremath{\mathinner{
\dfrac{\diff{^{#1}}#2}{\diff{{#3}^{#1}}}
}}}

\newcommand\bra[1]{\langle#1|}
\newcommand\ket[1]{|#1\rangle}
\newcommand\ketbra[2]{|#1\rangle\langle#2|}
\DeclarePairedDelimiter\norm{\|}{\|}

\newcommand{\defeq}{\coloneqq}

\DeclarePairedDelimiter\comtor{[}{]}

\DeclarePairedDelimiterX\setbuilder[2]{ \{ }{ \} }{#1\,\delimsize\vert\,\mathopen{}#2}

\title{Alternative adiabatic quantum dynamics with algorithmic applications}
\author{Joseph Cunningham\footnote{Centre for Quantum Information and Communication (QuIC), \'{E}cole polytechnique de Bruxelles,
Universit\'e libre de Bruxelles}\;\,\footnote{Laboratoire Bordelais de Recherche en Informatique (LaBRI), Université de Bordeaux} \and Jérémie Roland\footnotemark[1]}
\date{}

\begin{document}
\maketitle

\begin{abstract}
In adiabatic quantum computing the aim is to track an eigenstate as the Hamiltonian changes. In the usual setup this is achieved using the natural time-dependent Hamiltonian evolution of the system and the main technical tool is the adiabatic theorem. We propose several alternative processes that achieve the same goal, but can easily be implemented on a gate-based quantum computer without the overhead of simulating time-dependent Hamiltonian evolution. We give a general framework for deriving `adiabatic' theorems for these processes. 

As an application, we give various algorithms for solving the Quantum Linear Systems Problem (QLSP) with optimal scaling in the condition number. One of these algorithms was previously developed in \cite{cunninghamTQC} and another can be seen as a randomised version of the discrete adiabatic algorithm of \cite{QLSPdiscreteAdiabaticTheorem}.

We also describe versions of Trotterisation in our framework, which allows several results from \cite{an2025largetimestepdiscretisationadiabatic} to be reproduced in a randomised setting. In particular, bounds on the Trotter error in terms of the fidelity are obtained that are asymptotically better than the standard bounds. 
\end{abstract}

\section{Introduction}
Suppose $H(s)$ is a Hamiltonian for all $s\in [0,1]$ and we have an eigenstate of $H(0)$. The aim is to prepare a corresponding eigenstate of $H(1)$. In standard adiabatic quantum computing, this is achieved by allowing the system to evolve under a time-dependent Hamiltonian $H\big(s(t)\big)$, where the parameter $s$ is slowly changed from $0$ to $1$. The adiabatic theorem then guarantees that the state of the system remains close to the instantaneous eigenstate of $H(s)$. This has many algorithmic applications. For an overview, see \cite{Albash2018}.

In this paper we discuss several other procedures that track the eigenstate in a similar fashion for similar reasons, but can be implemented on a gate-based quantum computer. The usual way to port adiabatic algorithms to a digital quantum computer is by simulating the time-dependent Hamiltonian evolution, but this incurs a significant overhead \cite{GonzalezConde2025}. With our methods the same result can be achieved with no overhead: the complexity scales in the same way as the time complexity of the adiabatic algorithm.

This translation of adiabatic algorithms to the digital setting allows adiabatic algorithms to be used as subroutines in larger quantum computations. It is also enables the use of Hamiltonians that are not naturally able to be implemented on a physical quantum annealer. For example, the adiabatic version of Grover's algorithm makes use of one dimensional projectors \cite{farhiAQC, RolandLocalAQC, otherLocalAQC}. These are not local operators and are not typically natively supported by quantum annealers; their implementation on a gate-based device is more straightforward. The recent generalisation \cite{Braida2025} of Grover search to a larger class of optimisation problems has this same limitation and our methods give a way to implement the proposed algorithm. 

Similar considerations apply to the adiabatic version of the HHL algorithms, which solves the quantum Quantum Linear Systems Problem (QLSP) \cite{QLSPwithPR, An_QLSP, QLSPdiscreteAdiabaticTheorem}. It seems unfeasible to run this algorithm on a quantum annealer, but our methods use the adiabatic path to solve this problem optimally, in a fashion similar to \cite{QLSPdiscreteAdiabaticTheorem}.

The paper is structured as follows: first we introduce the basic definitions and objects we will be working with. In \autoref{sec:generalisedAdiabaticDynamics} we discuss the general mechanism for obtaining adiabatic dynamics; we apply it to the usual adiabatic setting with time-dependent Hamiltonian evolution and show how an adapted schedule can improve the scaling in the gap. In \autoref{sec:PoissonDistributedUnitaries} we design algorithms that have adiabatic behaviour due to a randomised application of unitaries; we develop results for three choices of unitary. In \autoref{sec:PoissonDistributedPhaseRandomisation} we replace the application of unitaries by phase randomisation. Finally, in \autoref{sec:applications}, we give two applications: Grover search and QLSP.

\subsection{Setup}
The general setup is illustrated in \autoref{fig:generalSetup}: we suppose a path of Hamiltonians $H(s)$, where $s\in [0,1]$. We assume that these Hamiltonians are bounded, but the Hilbert space on which they operate may be infinite-dimensional. This path is assumed twice continuously differentiable. We do not restrict ourselves to tracking a single eigenspace, rather we assume there exist two continuous functions $b_0(s), b_1(s)$ and attempt to track the eigenspace associated with the spectrum that lies in the interval $\big[b_0(s), b_1(s)\big]$, which we assume consists of a finite number of eigenvalues. The eigenspace of any of these eigenvalues may be degenerate. Let $P(s)$ denote the orthogonal projector on this space and set $Q(s) \defeq \mathbb{1} - P(s)$.

The distance between the spectrum inside $\big[b_0(s), b_1(s)\big]$ and the spectrum outside $\big[b_0(s), b_1(0)\big]$ is known as the \emph{gap}, denoted $g(s)$. This setup is fairly typical; see, e.g.\ \cite{Jansen2007}.

For notational clarity we often suppress the dependence on $s$. Primes always denote derivatives with respect to $s$.

We will also consider normal operators, whose spectrum does not necessarily lie on the line. In this case we replace $\big[b_0(s), b_1(s)\big]$ by a closed simple curve $\Gamma(s)$ that splits the spectrum into a part inside $\Gamma(s)$ and a part outside $\Gamma(s)$. The gap $g$ is the distance between these two sets. See \autoref{fig:spectrumIllustrationA}.

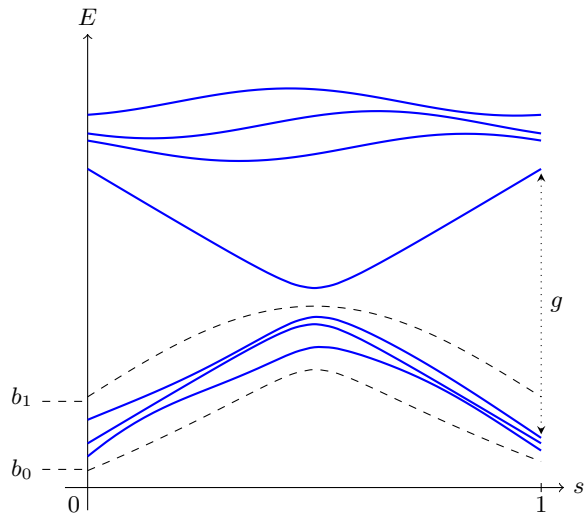
\begin{figure}[h!]
\centering
\begin{tikzpicture}[scale=6, every node/.style={font=\small}]

  \def\Eavg{0.4}      
  \def\Vval{0.04}     
  \def\kval{0.6}      
  
  \def\A{0.03}       

  \draw[->] (-0.05,0) -- (1.05,0) node[right] {$s$};
  \draw[->] (0,-0.05) -- (0,1) node[above] {$E$};
  
  \node[left] at (-0.1,0.04) {$b_0$};
  \node[left] at (-0.1,0.2) {$b_1$};
  \draw[dashed] (-0.1, 0.04) -- (0,0.04);
  \draw[dashed] (-0.1, 0.19) -- (0,0.19);
  
  \draw[<->, >=stealth, dotted] (1, { \Eavg - sqrt((\kval*(1-0.5))^2 + (\Vval)^2) + 0.02 }) -- (1, { \Eavg + sqrt((\kval*(1-0.5))^2 + (\Vval)^2) - 0.01 }) node[midway, right] {$g$};

  \node[below] at (-0.03,0) {0};
  \draw (1, -0.01) -- (1, 0.01);
  \node[below] at (1,0) {1};

  
  \draw[domain=0:1, smooth, variable=\x, blue, thick]
    plot ({\x}, { \Eavg - sqrt((\kval*(\x-0.5))^2 + (\Vval)^2) });
    
  \draw[domain=0:1, smooth, variable=\x, blue, thick] 
    plot ({\x}, {\Eavg - sqrt((\kval*(\x-0.5))^2 + (\Vval)^2) + (1.3-\x)*0.02*(1.5 + sin(360*\x + 150))});
  \draw[domain=0:1, smooth, variable=\x, blue, thick] 
    plot ({\x}, {\Eavg - sqrt((\kval*(\x-0.5))^2 + (\Vval)^2) - (1.5-\x)*0.02*(1.6 + sin(500*\x + 220))});
    
  \draw[domain=0:1, smooth, variable=\x, blue, thick]
    plot ({\x}, { \Eavg + sqrt((\kval*(\x-0.5))^2 + (\Vval)^2) });
    
  \draw[domain=0:1, smooth, variable=\x, dashed] plot ({\x}, {0.2+ 0.2*sin(180*\x) });

  \draw[domain=0:1, smooth, variable=\x, dashed]
    plot ({\x}, { \Eavg - sqrt((\kval*(\x-0.5))^2 + (\Vval)^2) - .02 - .2*\x*(1-\x) - (1-\x)*.04 - 0.02*\x });

  
  \draw[domain=0:1, smooth, variable=\x, blue, thick] 
    plot ({\x}, {0.75 + \A*sin(360*\x + 150)});
    
  \draw[domain=0:1, smooth, variable=\x, blue, thick] 
    plot ({\x}, {0.80 + \A*sin(360*\x + 220)});
    
  \draw[domain=0:1, smooth, variable=\x, blue, thick] 
    plot ({\x}, {0.85 + \A*sin(360*\x + 290)});

\end{tikzpicture}
\caption{An illustration of the spectrum of $H(s)$.} \label{fig:generalSetup}
\end{figure}

\section{Generalised adiabatic dynamics} \label{sec:generalisedAdiabaticDynamics}

In this work we consider various processes that have similar features: they are all capable of tracking an eigenspace in a certain parameter regime. The underlying mechanism is similar in all cases; the process is described by a differential equation that contains a multiplicative scalar parameter. When this parameter is taken to be large, the process tracks the eigenspace closely.

\subsection{A blueprint for adiabatic theorems}

Our procedures make use of (classical) randomness. For this reason, it will be convenient to work with density matrices, rather than state vectors. Our algorithms produce states $\rho(s)$, starting from an initial state $\rho(0)$. We want the fidelity of $\rho(s)$ with $P(s)$ to be larger, i.e.\ $\Tr(P(s)\rho(s)) \approx 1$.

We analyse our algorithms by showing that $\rho(s)$ can be obtained as the solution of a differential equation, with generator $\mathcal{L}_s$. This can be thought of as a Lindbladian. If this differential equation satisfies certain properties, then we show that the infidelity is small.

These properties are summarised in the following theorem, which forms the blueprint for all our adiabatic theorems.
\begin{theorem} \label{mainTheorem}
Let $\rho(s)$ be the solution of the differential equation
\begin{equation}
  \od{\rho}{s} = \lambda(s)\mathcal{L}_s(\rho), \label{eq:generalisedAdiabaticDifferentialQuation}
\end{equation}
where $\mathcal{L}_s$ is some $s$-dependent linear operator on density matrices.
Suppose
\begin{enumerate}
\item[(a)] $\Tr\big(P(0)\rho(0)\big) = 1$;
\item[(b)] $\Tr(P(s)Y) = 0$ for all $Y$ in the image of $\mathcal{L}_s$;
\item[(c)] there exists a differentiable $X(s)$ such that $P'(s) = \mathcal{L}^*_s\big(X(s)\big)$ for all $s\in [0,1]$;
\item[(d)] $\lambda^{-1}$ is absolutely continuous.
\end{enumerate}
Then
\begin{equation}
  1 - \Tr\big(P(1)\rho(1)\big) \leq \frac{\norm{X}}{\lambda}\bigg|_{s=0} + \frac{\norm{X}}{\lambda}\bigg|_{s=1} + \int_0^1 \Big(\frac{\norm{X'}}{\lambda} + \Big|\Big(\frac{1}{\lambda}\Big)'\Big|\norm{X}\Big)\diff{s}.
\end{equation}
\end{theorem}
Taking $\lambda$ large, the infidelity $1 - \Tr\big(P(1)\rho(1)\big)$ can be made arbitrarily small. The large $\lambda$ limit is the generalisation of the larger evolution time limit in the usual adiabatic setup.

Condition (a) just means that the state must initially be prepared in the eigenspace of interest. Condition (b) is required to make sure that any state in the eigenspace of interest is not pushed out by the dynamics. It is usually easy to check. Condition (c) is less easy to check. It requires some creativity to find a suitable $X$. The runtime of the algorithm depends on suitably bounding $\norm{X}$ and $\norm{X'}$. Condition (d) is a technicality; it holds whenever $\lambda^{-1}$ is continuous and piecewise continuously differentiable.

The proof of this result was inspired by a result in \cite{Avron2012}.
\begin{proof}[Proof of \autoref{mainTheorem}]
First calculate
\begin{align}
\Tr(P\rho)' &= \Tr(P'\rho) + \Tr(P\rho') \\
&= \Tr(P'\rho) + \lambda\Tr(P\mathcal{L}_s(\rho)) \\
&= \Tr(P'\rho) \\
&= \Tr\big(\mathcal{L}_s^*(X)\rho\big) \\
&= \Tr\big(X\mathcal{L}_s(\rho)\big) = \Tr(\lambda^{-1} X\rho').
\end{align}

Next we integrate this and use integration by parts to calculate
\begin{align}
1 - \Tr\big(P(1)\rho(1)\big) &= -\Tr(P\rho)|_0^1\\
&= -\int_0^1 \Tr(P\rho)' \diff{s} \\
&= -\int_0^1 \Tr(\lambda^{-1}X\rho') \diff{s} \\
&= -\lambda^{-1}\Tr(X\rho)|_0^1 \\
&\qquad + \int_0^1 \Big(\Tr(\lambda^{-1}X'\rho) + \Big(\frac{1}{\lambda}\Big)'\Tr(X\rho)\Big) \diff{s}.
\end{align}
Finally we use $\Tr(\rho) = 1$ and the triangle inequality to obtain
\begin{equation}
  1 - \Tr\big(P(1)\rho(1)\big) \leq \norm[\Big]{\frac{X}{\lambda}}\bigg|_{s=0} + \norm[\Big]{\frac{X}{\lambda}}\bigg|_{s=1} + \int_0^1 \Big(\frac{\norm{X'}}{\lambda} + \Big|\Big(\frac{1}{\lambda}\Big)'\Big|\norm{X}\Big)\diff{s}. 
\end{equation}
\end{proof}

\subsection{Rederiving the adiabatic theorem}
As a first application of \autoref{mainTheorem}, we derive an adiabatic theorem for time-dependent Hamiltonian evolution. In this case the evolution of the state is governed by the Liouville-von Neumann equation $\od{\rho}{t} = -i\comtor{H(s(t)), \rho}$. Reparametrising to $s$ with $\diff{t} = T\diff{s}$, we get
\begin{equation}
\dod{\rho}{s} = -iT \comtor{H(s), \rho}. \label{eq:LiouvilleVonNeumann}
\end{equation}
The parameter $T$ determines how slowly the Hamiltonian is changed. We identify $\mathcal{L}_s(\rho) = -i\comtor{H(s), \rho}$. Since $H$ and $P$ commute, it is immediate that $\Tr(P\comtor{H(s), \rho}) = 0$. In order to apply \autoref{mainTheorem}, we just need to find a suitable $X$.

For this, we make use of the so-called ``twiddle operation'', which is commonly used in proofs of the adiabatic theorem, e.g.\ \cite{Jansen2007}. This operation maps an operator $Y$ to an operator $\widetilde{Y}$ with the property that
\begin{equation}
\comtor{H,\widetilde{Y}} = \comtor{P,Y}. \label{eq:operatorEquation}
\end{equation}
This property alone allows us to construct a valid $X$. 
A precise definition and several results are given in \ref{sec:twiddleOperation}.

\begin{lemma} \label{timeAdiabaticX}
Let $H(s)$ be a path of bounded Hamiltonians that is twice continuously differentiable.

The differential equation \eqref{eq:LiouvilleVonNeumann} satisfies the assumptions of \autoref{mainTheorem} with
\begin{equation}
X = -i\widetilde{\comtor{P,P'}}.
\end{equation}
\end{lemma}
\begin{proof}
Let $\sigma$ be an arbitrary trace-one operator on $\mathcal{H}$. We need to prove that $\Tr\big(X \mathcal{L}_s(\sigma)\big) = \Tr(P'\sigma)$. Indeed,
\begin{align}
\Tr\big(X \mathcal{L}_s(\sigma)\big) &= -i\Tr\big(X \comtor{H,\sigma}\big) \\
&= -\Tr\Big(\widetilde{\comtor{P,P'}} \comtor{H,\sigma}\Big) \\
&= \Tr\Big(\comtor[\Big]{H, \widetilde{\comtor{P,P'}}}\sigma\Big) \\
&= \Tr\Big(\comtor[\big]{P, \comtor{P,P'}}\sigma\Big) = \Tr(P'\sigma).
\end{align}
\end{proof}

It is now straightforward to derive a version of the adiabatic theorem.

\begin{theorem}  \label{LiouvilleVonNeumannAdiabaticTheorem}
Let $H(s)$ be a path of bounded Hamiltonians that is twice continuously differentiable. Suppose $\sigma\big(H(s)\big)\cap \big[b_0(s), b_1(s)\big]$ consists of at most $m$ eigenvalues and $T^{-1}$ is absolutely continuous. The state that has been evolved under \eqref{eq:LiouvilleVonNeumann} has an infidelity bounded by
\begin{multline}
1 - \Tr\big(P(1)\rho(1)\big) \leq m\frac{\norm{H'}}{Tg^2}\Big|_{s=0} + m\frac{\norm{H'}}{Tg^2}\Big|_{s=1} + \int_0^1 \Big|\Big(\frac{1}{T}\Big)'\Big|m\frac{\norm{H'}}{g^2}\diff{s} \\
+ \int_0^1\bigg(\frac{m}{Tg^2}\norm{H''} + 5\frac{m\sqrt{m}}{Tg^3}\norm{H'}^2\bigg)\diff{s}.
\end{multline}
\end{theorem}
\begin{proof}
This is the result of an application of \autoref{mainTheorem}, using the $X$ from \autoref{timeAdiabaticX}. The bounds on $\norm{X}$ and $\norm{X'}$ come from \autoref{derivativeProjectorNormBound} and \autoref{boundingTimeAdiabaticX}.
\end{proof}

Comparing this theorem to the version of the adiabatic theorem as presented in \cite{Jansen2007}, we notice many similarities. One immediate difference is the inclusion of the term $\int_0^1 \big|\big(\frac{1}{T}\big)'\big|m\frac{\norm{H'}}{g^2}\diff{s}$, which is due to the fact that we do not assume that $T$ is constant. This simplifies the analysis of adapted schedules, like in \autoref{theorem:adaptiveRateLiouvilleAdiabaticTheorem}.

A more subtle difference is the slightly worse error dependence: the tightest versions of the adiabatic theorem bound the square root of the infidelity, rather than the infidelity.

On the other hand this theorem has a relatively simple proof and, more importantly, generalises to our other settings.

\subsection{Adapting the schedule}
Inspecting the adiabatic theorem, we see that, roughly speaking, the fidelity with the target state is high if $T$ is large compared to $g^{-3}$. As is illustrated in \autoref{fig:generalSetup}, $g(s)^{-3}$ is not constant. It stands to reason, then, that $T$ should not be taken constant either. In this section we outline a way to get optimal ``inverse gap'' scaling by choosing an adapted $T$, under some assumptions on the gap.

\begin{assumption} \label{hyp:integralOrderReduction}
There exists an absolutely continuous function $g_0: [0,1]\to \R$ such that
\begin{itemize}
\item $0< g_0(s) \leq g(s)$ for all $s\in [0,1]$;
\item there exists $p\in [1,2]$ and $B_{p},B_{3-p} \geq 0$ such that
\begin{equation}
  \int_0^1 \frac{1}{g_0(s)^p} \diff{s} \leq B_{p} g_{0m}^{1-p} \quad\text{and}\quad \int_0^1 \frac{1}{g_0(s)^{3-p}} \diff{s} \leq B_{3-p} g_{0m}^{p-2},
\end{equation}
where $g_{0m} = \min_{s\in[0,1]}g_0(s)$.
\end{itemize}
\end{assumption}

This assumption holds for many interesting paths of Hamiltonians

\begin{theorem} \label{theorem:adaptiveRateLiouvilleAdiabaticTheorem}
Let $H(s)$ be a path of bounded Hamiltonians that is twice continuously differentiable. Suppose $\sigma(H(s))\cap [b_0(s), b_1(s)]$ consists of at most $m$ eigenvalues and \autoref{hyp:integralOrderReduction} holds.
Take $C\geq 0$ such that
\begin{equation}
  C \geq m\max_{s\in [0,1]}\Big(\big(2+p|g_0'|B_{3-p}\big)\norm{H'} + \norm{H''} + 5\sqrt{m}B_{3-p}\norm{H'}^2\Big).
\end{equation}
Fix $\epsilon > 0$ and set
\begin{equation}
  T = \frac{1}{\epsilon} \frac{C}{g_0^pg_{0m}^{2-p}}.
\end{equation}
Then the evolved state has fidelity $\Tr\big(P(1)\rho(1)\big) \geq 1-\epsilon$. The total evolution time satisfies
\begin{equation}
  \frac{1}{\epsilon}\frac{CB_{p}}{g_{0m}}. \label{eq:adaptiveRateLiouvilleAdiabaticTheoremTotalTime}
\end{equation}
\end{theorem}
\begin{proof}
Two claims need to be proved: firstly that $1-\Tr\big(P(1)\rho(1)\big) \leq \epsilon$ and secondly that the algorithm finishes in a time bounded by \eqref{eq:adaptiveRateLiouvilleAdiabaticTheoremTotalTime}.

For the first claim, \autoref{LiouvilleVonNeumannAdiabaticTheorem} gives a bound on $1-\Tr\big(P(1)\rho(1)\big)$. This theorem can be applied since $x^{p}$ is Lipschitz on $\big[g_{0m}, \max_s{\norm{H(s)}}\big]$, so $T^{-1}$ is absolutely continuous. We just need to show that it evaluates to something smaller than $\epsilon$ in this case. First observe that
\begin{equation}
  \Big|\Big(\frac{1}{T}\Big)'\Big| = \frac{\epsilon g_{0m}^{2-p}}{C} pg_0^{p-1}|g_0'|.
\end{equation}
Then the claim follows from the following calculation:
\begin{align}
1-\Tr\big(P(1)\rho(1)\big) &\leq m\frac{\epsilon}{C}\bigg(\frac{\norm{H'}}{g_{0m}^{p-2}g_0^{-p}g^2}\bigg|_{s=0} + \frac{\norm{H'}}{g_{0m}^{p-2}g_0^{-p}g^2}\bigg|_{s=1} + \int_0^1 p|g_0'|g_{0m}^{2-p}g_0^{p-1}\frac{\norm{H'}}{g^2}\diff{s} \nonumber\\
&\hspace{3em}+ \int_0^1\bigg(\frac{\norm{H''}}{g_{0m}^{p-2}g_0^{-p}g^2} + 5\sqrt{m}\frac{\norm{H'}^2}{g_{0m}^{p-2}g_0^{-p}g^3}\bigg)\diff{s}\bigg) \\
&\leq m\frac{\epsilon}{C}\bigg(\norm{H'}\Big|_{s=0} + \norm{H'}\Big|_{s=1} + \int_0^1 p|g_0'|\frac{\norm{H'}}{g_{0m}^{p-2}g_0^{3-p}}\diff{s} \nonumber\\
&\hspace{3em}+ \int_0^1\bigg(\norm{H''} + 5\sqrt{m}\frac{\norm{H'}^2}{g_{0m}^{p-2}g_0^{3-p}}\bigg)\diff{s}\bigg) \\
&\leq m\frac{\epsilon}{C}\max_{s\in [0,1]}\Big(\big(2+p|g_0'|B_{3-p}\big)\norm{H'} + \norm{H''} + 5\sqrt{m}B_{3-p}\norm{H'}^2\Big) \\
&\leq \epsilon.
\end{align}
Finally the time complexity of the algorithm is bounded by
\begin{align*}
\int_0^1 T\diff{s} &= \frac{C}{\epsilon g_{0m}^{2-p}}\int_0^1 \frac{1}{g_0^{p}}\diff{s} \\
&\leq \frac{1}{\epsilon}\frac{CB_{p}}{g_{0m}}.
\end{align*}
\end{proof}

\section{Poisson-distributed unitaries} \label{sec:PoissonDistributedUnitaries}
We now turn to other processes that we want to use to track eigenstates. As a first example, suppose we have a path of unitaries $U(s)$ and we apply this unitary to the state for $s$ in some finite subset of $[0,1]$. In other words, fix $\{s_0,\ldots, s_m\} \subseteq [0,1]$ and apply the unitary $U(s_k)$ to the system, for each $s_k$ in order.

It is known that this procedure tracks the eigenstate if the set of discretisation points is chosen sufficiently dense. This is exactly the discrete adiabatic theorem, which is originally due to \cite{Dranov1998} and was extended and refined in \cite{QLSPdiscreteAdiabaticTheorem}.
The advantage of this procedure is that it can easily be performed on a discrete quantum computer, but there is a problem when trying to analyse it in our framework: this discrete procedure will not be described by a differential equation.

We resolve this by adding some randomness; rather than picking the discretisation schedule deterministically, we choose it stochastically as the realisation of some (non-homogeneous) Poisson process. Such a realisation is illustrated in \autoref{fig:PoissonProcess}. The crosses on the horizontal axis show which points were chosen in this realisation. It is conventional to identify the Poisson process with the function $N$ that jumps up by one at each point in the realisation.

\begin{figure}[h!]
\centering
\begin{tikzpicture}[xscale=0.8, yscale=0.7]
  \def\offsets{1, .3, .7, 2, .1, .3, .1, 3, 1, 2.1, .5, .2, .3}
  
  \draw[->, thick] (-1,0) -- (13,0);
  \draw[->, thick] (0,-.5) -- (0,7);
  
  \draw \foreach \x in \offsets{
  ++(\x,0) node {$\times$}
  };
  
  \draw (0,0) \foreach \x in \offsets{
  -- ++(\x,0) circle (.05) ++ (0,.5)
  } -- ++(1,0);
  \draw[fill=black] (0,0) \foreach \x in \offsets{
  ++(\x,.5) circle (.05)
  };
  
  \draw (-.5, 6.5) node {$N$};
\end{tikzpicture}
\caption{A realisation of a Poisson process.} \label{fig:PoissonProcess}
\end{figure}
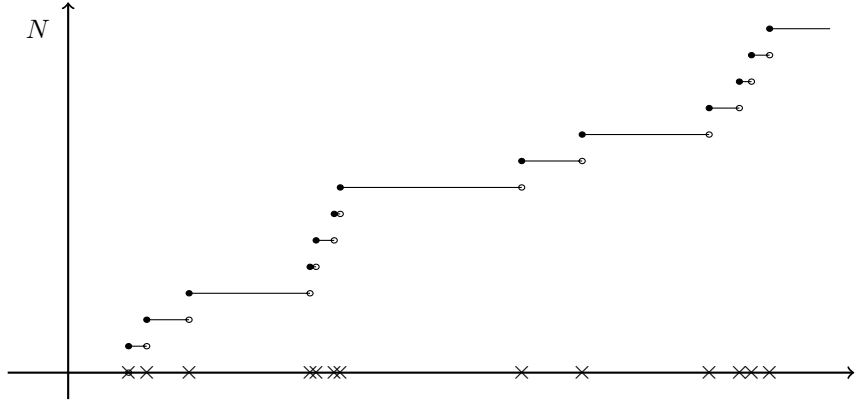

Operationally, the procedure becomes \autoref{procedureDiscreteAdiabaticTheorem}.

\begin{algorithm}
Pick a Poisson process $N: [0,1] \times (\Omega, \mathcal{A}, P)\to \mathbb{N}$ with rate $\lambda(s)$\;
At each jump point $s$ of the Poisson process, apply $U(s)$\;
\caption{Poisson-distributed unitaries.} \label{procedureDiscreteAdiabaticTheorem}
\end{algorithm}

This process is described by the stochastic differential equation $\diff{\rho} = \big(U\rho U^* - \rho\big)\diff{N}$.
Marginalising over the realisations (and applying Campbell's theorem) gives the (non-stochastic) differential equation
\begin{equation}
  \od{\rho}{s} = \lambda\big(U\rho U^* - \rho\big). \label{eq:discreteUnitaryDynamics}
\end{equation}

We have now obtained a differential equation and it satisfies the assumptions of \autoref{mainTheorem}.

\begin{lemma} \label{discreteAdiabaticX}
Let $U(s)$ be a path of unitaries that is twice continuously differentiable.

The differential equation \eqref{eq:discreteUnitaryDynamics} satisfies the assumptions of \autoref{mainTheorem} with
\begin{equation}
X = \widetilde{\comtor{UP', P}}.
\end{equation}
\end{lemma}
Here $P$ is an eigenprojector of $U$ and The twiddle is with respect to $U$, meaning it has the property that $\comtor{U,\widetilde{Y}} = \comtor{P,Y}$.
\begin{proof}
Let $\sigma$ be an arbitrary trace-one operator on $\mathcal{H}$. We need to prove that $\Tr\big(X \mathcal{L}_s(\sigma)\big) = \Tr(P'\sigma)$. Indeed, using $\comtor{X, U} = UP'$, we have
\begin{align}
\Tr\big(X \mathcal{L}_s(\sigma)\big) &= \Tr\big(X (U\sigma U^* - \sigma)\big) \\
&= \Tr\big(\comtor{X, U} \sigma U^*\big) \\
&= \Tr\big(U P' \sigma U^*\big) = \Tr(P'\sigma).
\end{align}
\end{proof}

\begin{theorem}  \label{discreteAdiabaticTheorem}
Let $U(s)$ be twice continuously differentiable. Suppose the part of $\sigma\big(U(s)\big)$ inside $\Gamma(s)$ consists of at most $m$ eigenvalues and $\lambda^{-1}$ is absolutely continuous. The state that has been evolved under \eqref{eq:discreteUnitaryDynamics} has an infidelity bounded by
\begin{align}
1 - \Tr\big(P(1)\rho(1)\big) &\leq \sqrt{m}\frac{\norm{P'}}{\lambda g}\Big|_{s=0} + \sqrt{m}\frac{\norm{P'}}{\lambda g}\Big|_{s=1} + \int_0^1 \Big|\Big(\frac{1}{\lambda}\Big)'\Big|\sqrt{m}\frac{\norm{P'}}{g}\diff{s} \nonumber \\
&\qquad+ \int_0^1\bigg(\frac{\sqrt{m}}{\lambda g}\norm{U'}\,\norm{P'} + 2\frac{m}{\lambda g^2}\norm{U'}\,\norm{P'} +  \frac{\sqrt{m}}{\lambda g}\norm{[P', P]'} + \frac{\sqrt{m}}{\lambda g}\norm{P'}^2\bigg)\diff{s} \\
&\leq m\frac{\norm{U'}}{\lambda g^2}\Big|_{s=0} + m\frac{\norm{U'}}{\lambda g^2}\Big|_{s=1} + \int_0^1 \Big|\Big(\frac{1}{\lambda}\Big)'\Big|m\frac{\norm{U'}}{g^2}\diff{s} \nonumber \\
&\qquad+ \int_0^1\bigg(\frac{m}{\lambda g^2}\big(\norm{U'}^2 + \norm{U''}\big) + 5\frac{m\sqrt{m}}{\lambda g^3}\norm{U'}^2\bigg)\diff{s}.
\end{align}
\end{theorem}
There are two bounds on $1 - \Tr\big(P(1)\rho(1)\big)$ stated in this theorem. The tighter bound is not particularly useful on its own, since it still contains derivatives of $P$. However, we will consider applications where the projector $P$ has some relation to a path of Hamiltonians and we will be able to leverage this connection to give a tighter bound.
\begin{proof}
This is the result of an application of \autoref{mainTheorem}, using the $X$ from \autoref{discreteAdiabaticX}. We can bound $\norm{X} \leq \frac{\sqrt{m}}{g}\norm{P'}$ and
\begin{equation}
\norm{X'} \leq \frac{\sqrt{m}}{g}\norm{U'}\,\norm{P'} + 2\frac{m}{g^2}\norm{U'}\,\norm{P'} +  \frac{\sqrt{m}}{g}\norm{[P', P]'} + \frac{\sqrt{m}}{g}\norm{P'}^2
\end{equation}
using \autoref{normBoundXTildeFiniteSetEigenvalues} and \autoref{derivativeOfTwiddle} (with the observation that $\comtor[\big]{P', \comtor{UP', P}}$ is diagonal and so performing the twiddle operation maps it to zero).

The second inequality follows from \autoref{derivativeProjectorNormBound} and \autoref{boundingTimeAdiabaticX}.
\end{proof}

As before, we can also consider an adapted schedule.

\begin{theorem} \label{theorem:adaptiveRateDiscreteAdiabaticTheorem}
Let $U(s)$ be twice continuously differentiable. Suppose the part of $\sigma\big(U(s)\big)$ inside $\Gamma(s)$ consists of at most $m$ eigenvalues and \autoref{hyp:integralOrderReduction} holds.
Take $C\geq 0$ such that
\begin{equation}
  C \geq m\max_{s\in [0,1]}\Big(\big(2+p|g_0'|B_{3-p}\big)\norm{U'} + \norm{U''} + \norm{U'}^2 + 5\sqrt{m}B_{3-p}\norm{U'}^2\Big).
\end{equation}
Fix $\epsilon > 0$ and set
\begin{equation}
  \lambda = \frac{1}{\epsilon} \frac{C}{g_0^pg_{0m}^{2-p}}
\end{equation}
Then the evolved state has fidelity $\Tr\big(P(1)\rho(1)\big) \geq 1-\epsilon$. The total expected number of applications of the unitary is
\begin{equation}
  \int_0^1 \lambda \diff{s} = \frac{1}{\epsilon}\frac{CB_{p}}{g_{0m}}. \label{eq:adaptiveRateDiscreteAdiabaticTheoremTotalTime}
\end{equation}
\end{theorem}
The proof is completely analogous to that of \autoref{theorem:adaptiveRateLiouvilleAdiabaticTheorem}.

In the following sections we consider various choices for $U(s)$ and derive adiabatic theorems in these cases.

\subsection{Qubitised unitaries}
Consider a path of Hamiltonians $H(s)$ with $\norm{H} \leq 1$. Then we can consider the qubitisation
\[ U(s) = \begin{pmatrix} H & -\sqrt{\mathbb{1} - H^2} \\ \sqrt{\mathbb{1} - H^2} & H\end{pmatrix}, \]
which is a unitary operator. For information on how to implement this operator, see \cite{Low2019}. Since we are considering both a path of Hamiltonians $H(s)$ and a path of unitaries $U(s)$ at the same time, our notation is somewhat ambiguous: the projector $P$ and gap $g$ could belong to either. In such cases we always consider $P$ and $g$ to be those of the path of Hamiltonians $H(s)$.

Setting $\ket{y+} \defeq \frac{1}{\sqrt{2}}\begin{pmatrix}1 \\ i\end{pmatrix}$ and $\ket{y-} \defeq \frac{1}{\sqrt{2}}\begin{pmatrix}1 \\ -i\end{pmatrix}$, it is straightforward to verify that
\begin{equation}
U = \ketbra{y-}{y-}\otimes \big(H + i\sqrt{\mathbb{1}-H^2}\big) + \ketbra{y+}{y+}\otimes \big(H - i\sqrt{\mathbb{1}-H^2}\big).
\end{equation}
From this it is immediate that the spectrum of $U$ is $\setbuilder{\alpha \pm \sqrt{1-\alpha^2}}{\alpha\in \sigma(H)}$. The spectrum has doubled and the gaps in the spectrum have increased:
\begin{align}
\Big|(\omega_0 \pm i\sqrt{1-\omega_0^2}) - (\omega_1 \pm i\sqrt{1-\omega_1^2})\Big| &= \sqrt{(\omega_0-\omega_1)^2 + \big(\sqrt{1-\omega_0^2} \pm \sqrt{1-\omega_1^2}\big)} \\
&\geq |\omega_0 - \omega_1|.
\end{align}
Define $P_+ = \ketbra{y-}{y-}\otimes P$ and $P_- = \ketbra{y+}{y+}\otimes P$. We want to track the eigenspace associated with $P_+$.

\begin{theorem}
Let $H(s)$ be a path of bounded Hamiltonians that is twice continuously differentiable and satisfies $\norm{H(s)} \leq 1$ for all $s\in [0,1]$. Suppose $\sigma\big(H(s)\big)\cap \big[b_0(s), b_1(s)\big]$ consists of at most $m$ eigenvalues and $\lambda^{-1}$ is absolutely continuous. Then
\begin{multline}
1 - \Tr\big(P_+(1)\rho(1)\big) \leq m\frac{\norm{H'}}{\lambda g^2}\Big|_{s=0} + m\frac{\norm{H'}}{\lambda g^2}\Big|_{s=1} + \int_0^1 \Big|\Big(\frac{1}{\lambda}\Big)'\Big|m\frac{\norm{H'}}{g^2}\diff{s} \\
+ \int_0^1\bigg(\Big(1+\frac{1}{\sqrt{1-\norm{H}^2}}\Big)\frac{m}{\lambda g^2}\norm{H'}^2 + \frac{m}{\lambda g^2}\norm{H''} + \Big(5+\frac{2}{\sqrt{1-\norm{H}^2}}\Big)\frac{m\sqrt{m}}{\lambda g^3}\norm{H'}^2\bigg)\diff{s}.
\end{multline}
\end{theorem}
\begin{proof}
This is an application of \autoref{discreteAdiabaticTheorem}.
We have $P_+' = \ketbra{y-}{y-}\otimes P'$.
\autoref{derivativeSquareRoot} gives
\begin{equation}
\sqrt{\mathbb{1} - H^2}' = \int_0^\infty e^{-u\sqrt{\mathbb{1}-H^2}}\big(H' H + HH'\big)e^{-u\sqrt{\mathbb{1}-H^2}} \diff{u}.
\end{equation}
Taking the norm, we bound
\begin{align}
\norm{\sqrt{\mathbb{1} - H^2}'} &\leq 2\norm{H'}\int_0^\infty e^{-2u\sqrt{1-\norm{H}^2}} \diff{u} \\
&= \frac{\norm{H'}}{\sqrt{1-\norm{H}^2}}.
\end{align}
Now we have
\begin{equation}
\norm{U'} \leq \norm*{\begin{pmatrix}\norm{H'} & \norm{\sqrt{\mathbb{1}-H^2}'} \\ \norm{\sqrt{\mathbb{1}-H^2}'} & \norm{H'} \end{pmatrix}} = \norm{H'} + \norm{\sqrt{\mathbb{1}-H^2}'} \leq \Big(1+\frac{1}{\sqrt{1-\norm{H}^2}}\Big)\norm{H'}.
\end{equation}
We conclude with \autoref{derivativeProjectorNormBound} and \autoref{boundingTimeAdiabaticX}.
\end{proof}

As before, we can consider an adapted schedule.

\begin{theorem} \label{theorem:adaptiveRateQubitisedDiscreteAdiabaticTheorem}
Let $H(s)$ be a path of bounded Hamiltonians that is twice continuously differentiable. Suppose $\sigma\big(H(s)\big)\cap \big[b_0(s), b_1(s)\big]$ consists of at most $m$ eigenvalues and \autoref{hyp:integralOrderReduction} holds.
Take $C\geq 0$ such that
\begin{equation}
  C \geq m\max_{s\in [0,1]}\bigg(\big(2+p|g_0'|B_{3-p}\big)\norm{H'} + \norm{H''} + \Big(1 + \frac{1}{\sqrt{1-\norm{H}^2}}\Big)\norm{H'}^2 + \Big(5 + \frac{2}{\sqrt{1-\norm{H}^2}}\Big)\sqrt{m}B_{3-p}\norm{H'}^2\bigg).
\end{equation}
Fix $\epsilon > 0$ and set
\begin{equation}
  \lambda = \frac{1}{\epsilon} \frac{C}{g_0^pg_{0m}^{2-p}}
\end{equation}
Then the evolved state has fidelity $\Tr\big(P(1)\rho(1)\big) \geq 1-\epsilon$. The total expected number of applications of the unitary is
\begin{equation}
  \int_0^1 \lambda \diff{s} = \frac{1}{\epsilon}\frac{CB_{p}}{g_{0m}}.
\end{equation}
\end{theorem}
The proof is completely analogous to that of \autoref{theorem:adaptiveRateLiouvilleAdiabaticTheorem}.

\subsection{Steps of time-independent evolution}
We again assume $\norm{H(s)} \leq \frac{1}{2}$ and consider the unitary $U(s) = e^{-i\frac{\pi}{2}H(s)}$. Note that the eigenspaces of $U(s)$ are the same as the eigenspaces of $H(s)$.

Jordan's inequality gives
\begin{equation}
|e^{-i\frac{\pi}{2}\omega_0} - e^{-i\frac{\pi}{2}\omega_1}| = |1- e^{-i\frac{\pi}{2}(\omega_1 - \omega_0)}| \geq \sin\Big(\frac{\pi}{2}|\omega_1-\omega_0|\Big) \geq |\omega_1-\omega_0|,
\end{equation}
for any $\omega_0, \omega_1 \in [-2^{-1},2^{-1}]$, so the gap of $U(s)$ is lower bounded by the gap $g$ of $H(s)$. Now we are ready to state and prove the adiabatic theorem.

\begin{theorem}
Let $H(s)$ be a path of bounded Hamiltonians that is twice continuously differentiable and satisfies $\norm{H(s)} \leq \frac{1}{2}$ for all $s\in [0,1]$. Suppose $\sigma\big(H(s)\big)\cap \big[b_0(s), b_1(s)\big]$ consists of at most $m$ eigenvalues and $\lambda^{-1}$ is absolutely continuous. Then
\begin{multline}
1 - \Tr\big(P(1)\rho(1)\big) \leq m\frac{\norm{H'}}{\lambda g^2}\Big|_{s=0} + m\frac{\norm{H'}}{\lambda g^2}\Big|_{s=1} + \int_0^1 \Big|\Big(\frac{1}{\lambda}\Big)'\Big|m\frac{\norm{H'}}{g^2}\diff{s} \\
+ \int_0^1\bigg(\frac{\pi}{2}\frac{m}{\lambda g^2}\norm{H'}^2 + \frac{m}{\lambda g^2}\norm{H''} + (3+\pi)\frac{m\sqrt{m}}{\lambda g^3}\norm{H'}^2\bigg)\diff{s}.
\end{multline}
\end{theorem}
\begin{proof}
The eigenprojector of $U(s)$ is the same as the eigenprojector of $H(s)$.

We have, from \autoref{derivativeExponential}, that
\begin{equation}
(e^{-i\frac{\pi}{2}H})' = -i\frac{\pi}{2}\int_0^\infty e^{-i\frac{\pi}{2}(1-u)H}H'e^{-i\frac{\pi}{2}uH} \diff{u},
\end{equation}
and taking the norm gives
\begin{equation}
\norm{(e^{-i\frac{\pi}{2}H})'} \leq \frac{\pi}{2}\norm{H'}.
\end{equation}
Using this norm bound, as well as the ones in \autoref{derivativeProjectorNormBound} and \autoref{boundingTimeAdiabaticX}, the result follows from \autoref{discreteAdiabaticTheorem}.
\end{proof}

\begin{theorem} \label{theorem:adaptiveRateTimeIndependentDiscreteAdiabaticTheorem}
Let $H(s)$ be a path of bounded Hamiltonians that is twice continuously differentiable and satisfies $\norm{H(s)} \leq \frac{1}{2}$ for all $s\in [0,1]$. Suppose $\sigma\big(H(s)\big)\cap \big[b_0(s), b_1(s)\big]$ consists of 
at most $m$ eigenvalues and \autoref{hyp:integralOrderReduction} holds.
Take $C\geq 0$ such that
\begin{equation}
  C \geq m\max_{s\in [0,1]}\bigg(\big(2+p|g_0'|B_{3-p}\big)\norm{H'} + \norm{H''} + \frac{\pi}{2}\norm{H'}^2 + \big(3+\pi\big)\sqrt{m}B_{3-p}\norm{H'}^2\bigg).
\end{equation}
Fix $\epsilon > 0$ and set
\begin{equation}
  \lambda = \frac{1}{\epsilon} \frac{C}{g_0^pg_{0m}^{2-p}}
\end{equation}
Then the evolved state has fidelity $\Tr\big(P(1)\rho(1)\big) \geq 1-\epsilon$. The total expected number of applications of the unitary is
\begin{equation}
  \int_0^1 \lambda \diff{s} = \frac{1}{\epsilon}\frac{CB_{p}}{g_{0m}}.
\end{equation}
\end{theorem}
The proof is completely analogous to that of \autoref{theorem:adaptiveRateLiouvilleAdiabaticTheorem}.

\subsection{Elementary Trotter steps}
Suppose $H(s) = (1-s)H_0 + sH_1$, for some fixed Hamiltonians $H_0, H_1$ with $\norm{H_0}, \norm{H_1} \leq 1$. We consider a first-order Trotter step
\begin{equation}
U_{1,h}(s) \defeq e^{-i\frac{\pi}{2}(1-s)hH_0}e^{-i\frac{\pi}{2}shH_1}
\end{equation}
and a second-order Trotter step
\begin{equation}
U_{2,h}(s) \defeq e^{-i\frac{\pi}{4}(1-s)hH_0}e^{-i\frac{\pi}{2}shH_1}e^{-i\frac{\pi}{4}(1-s)hH_0},
\end{equation}
for any $h \geq 0$. Since $U_{1,h}(s) = e^{-i\frac{\pi}{4}(1-s)hH_0}U_{2,h}(s)e^{i\frac{\pi}{4}(1-s)hH_0}$, it is clear that $U_{1,h}(s)$ and $U_{2,h}(s)$ have the same spectrum.

Things are trickier in this situation because the eigenspaces of $U_{1,h}(s)$ and $U_{2,h}(s)$ do not coincide with those of $H(s)$ (unless $s = 0,1$).
It is a standard application of the Neumann series expansion to observe that adding a perturbation to an operator can only shift the spectrum by at most the norm of the perturbation. We can use this to obtain a lower bound on the gap.

We have the following inequalities from \cite{suzuki_decomposition_1985}:
\begin{align}
\norm[\big]{U_{2,h}(s) - e^{-i\frac{\pi}{2}h\big((1-s)H_0 + sH_1\big)}} &\leq \frac{\pi^3h^3(1-s)s^2}{96} \norm[\big]{\comtor[\big]{H_1, \comtor{H_1, H_0}}} + \frac{\pi^3h^3(1-s)^2s}{192} \norm[\big]{\comtor[\big]{H_0, \comtor{H_0, H_1}}} \\
&\leq \frac{\pi^3h^3(1-s)s^2}{24} + \frac{\pi^3h^3(1-s)^2s}{48} \leq \frac{h^3}{2},
\end{align}
where we have used $(1-s)s \leq 0.25$. This implies that the gap of $U_{1,h}(s)$ is greater than $h(g - h^2)$.
Now we can derive an adiabatic theorem.

\begin{theorem}
Let $H(s) = (1-s)H_0 + sH_1$, with $\norm{H_0}, \norm{H_1} \leq 1$. Suppose $\sigma\big(H(s)\big)\cap \big[b_0(s), b_1(s)\big]$ consists of at most $m$ eigenvalues and that $0 < h < \sqrt{g}$. Then
\begin{multline}
1 - \Tr\big(P(1)\rho(1)\big) \leq m\frac{\pi}{2}\frac{\norm{H_1 - H_0}}{\lambda h(g - h^2)^2}\Big|_{s=0} + m\frac{\pi}{2}\frac{\norm{H_1 - H_0}}{\lambda h(g - h^2)^2}\Big|_{s=1} \\ 
+ \norm{H_1-H_0}\int_0^1\bigg(\frac{m}{\lambda (g-h^2)^2}\pi^2 + 5\frac{m\sqrt{m}}{\lambda h(g - h^2)^3} \frac{\pi^2}{4}\norm{H_1-H_0} + \Big|\Big(\frac{1}{\lambda}\Big)'\Big|m\frac{\pi}{2}\frac{1}{h(g-h^2)^2}\bigg)\diff{s}.
\end{multline}
\end{theorem}
\begin{proof}
We have
\begin{equation}
(e^{-i\frac{\pi}{2}(1-s)hH_0}e^{-i\frac{\pi}{2}shH_1})' = ih\frac{\pi}{2}e^{-i\frac{\pi}{2}(1-s)hH_0}(H_0 - H_1)e^{-i\frac{\pi}{2}shH_1},
\end{equation}
and taking the norm gives
\begin{equation}
\norm{U(s)'} \leq h\frac{\pi}{2}\norm{H_1-H_0}.
\end{equation}
Similarly,
\begin{equation}
\norm{U(s)''} \leq h^2 \frac{\pi^2}{4}\big(\norm{H_0} + \norm{H_1}\big)\norm{H_1-H_0} \leq h^2 \frac{\pi^2}{2}\norm{H_1-H_0}.
\end{equation}
The result now follows from \autoref{discreteAdiabaticTheorem}.
\end{proof}

In order to get a good fidelity, we need to take  $\lambda =  \Omega\big(\frac{1}{h g^3}\big) = \Omega\big(\frac{1}{g^{3.5}}\big)$. This is a worse scaling than the previous methods, but a better scaling than you get from applying the usual Trotter bound and thus shows that the usual analysis of Trotterisation is not tight in an adiabatic setting. This is essentially the same result as that of \cite{an2025largetimestepdiscretisationadiabatic} and we refer to this paper for a more in-depth discussion.

\section{Poisson-distributed phase randomisation} \label{sec:PoissonDistributedPhaseRandomisation}

For this next procedure, we also perform some discrete action according to a Poisson process, but rather than apply unitaries, we perform phase randomisation. This means that we evolve the system under the fixed, time-independent, Hamiltonian $H(s)$ for some random amount of time. If we choose a good distribution for the random amount of time, then this effectively implements a kind of projection onto the eigenspace we wish to track.

This is made more concrete in the following result.
\begin{proposition}[Phase randomisation] \label{phaseRandomisation}
Let $H$ be a Hamiltonian, $\omega_0$ an isolated point in the spectrum, $P$ the projector on the associated eigenspace and $g_0$ a lower bound on the spectral gap. Assume we can simulate $e^{-itH}$ for any positive or negative time $t$ at a cost of $|t|$. Then we can construct a stochastic variable $\tau$, with distribution $\mu$ such that for all states $\rho$,
\begin{equation}
\int_{-\infty}^\infty e^{-i\tau H}\rho e^{i\tau H}\diff{\mu(\tau)} = P\rho P + Q\Big(\int_{-\infty}^\infty e^{-i\tau H}\rho e^{i\tau H} \diff{\mu(\tau)}\Big) Q,
\end{equation}
with $\int_{-\infty}^\infty |\tau| \diff{\mu(\tau)} = t_0/g_0$, where $t_0 = 2.32132$. 
\end{proposition}
The result is originally from \cite{eigenpathTraversalPR}. The value for $t_0$ was obtained in \cite{QLSPwithPRnumericsv2}.

As before, we perform this phase randomisation operation at points specified by a Poisson process. Operationally, this amounts to \autoref{procedure}.

\begin{algorithm}
Pick a Poisson process $N: [0,1] \times (\Omega, \mathcal{A}, P)\to \mathbb{N}$ with rate $\lambda(s)$\;
At each jump point $s$ of the Poisson process, pick an instance $t$ of the random variable $\tau$ as defined in \autoref{phaseRandomisation} and evolve the system under the Hamiltonian evolution $e^{-itH(s)}$\;
\caption{Poisson-distributed phase randomisation.} \label{procedure}
\end{algorithm}

The density matrix describing the system is a random variable that satisfies the stochastic differential equation $\diff{\rho} = \big(e^{-i\tau(s)H(s)}\rho e^{i\tau(s)H(s)}- \rho\big)\diff{N}.$

Marginalising over the Poisson process (i.e.\ forgetting which exact realisation of the Poisson process was picked) and over the stochastic variable $\tau$, we get a new density matrix that is determined by the following differential equation:
\begin{equation} \dod{\rho}{s} = \lambda\Big(P\rho P + \int_{-\infty}^\infty Q e^{-i\tau H(s)}\rho e^{i\tau H(s)} Q \diff{\mu(\tau)} - \rho\Big), \label{eq:phaseRandomisationDynamics}
\end{equation}
with probability distribution $\mu$. This result is basically an application of Campbell's theorem, but there is a slight technicality: we are now averaging over both the Poisson process and the stochastic variable $\tau$, so in principle we need a stronger version of Campbell's theorem. Deriving this stronger version is not too difficult, however. Like the original theorem it is a straightforward application of the dominated convergence theorem.

The total time taken by one run of the algorithm is a random variable $T$ satisfying $\diff{T} = \tau\diff{N}$.
In order to find the time complexity, we again marginalise over the Poisson realisations. This gives
\begin{equation}
  \diff{T} = g^{-1}\lambda\diff{s},
\end{equation}
so $T = \int_0^1\frac{\lambda}{g}\diff{s}$.

\begin{lemma}
Let $H(s)$ be a path of bounded Hamiltonians that is twice continuously differentiable. The differential equation \eqref{eq:phaseRandomisationDynamics} satisfies the assumptions of \autoref{mainTheorem} with
\begin{equation}
X(s) = -P'(s)
\end{equation}
\end{lemma}
\begin{proof}
First observe that any operator in the image of $\mathcal{L}_s$ is off-diagonal, so multiplying by $P$ and taking the trace gives zero.

Let $\sigma \in \Bounded_1(\mathcal{H})$ be an arbitrary trace class operator. We need to prove that $\Tr(X \mathcal{L}_s(\sigma)) = \Tr(P'\sigma)$. Since $P'$ is off-diagonal, it is clear that
\begin{align}
\Tr(P'\sigma) &= -\Tr\bigg(P'\Big(P\rho P + \int Q e^{-i\tau H(s)}\rho e^{i\tau H(s)} Q \diff{\mu(\tau)}-\sigma\Big)\bigg) \\
&= \Tr\big((-P')\mathcal{L}_s(\sigma)\big).
\end{align}
\end{proof}

\begin{theorem}  \label{poissonDistributedPhaseRandomisationAdiabaticTheorem}
Let $H(s)$ be a path of bounded Hamiltonians that is twice continuously differentiable. Suppose $\sigma(H(s))\cap [b_0(s), b_1(s)]$ consists of exactly one eigenvalue and $\lambda^{-1}$ is absolutely continuous. The dynamics of \eqref{eq:phaseRandomisationDynamics} produces a state with infidelity bounded by
\begin{align}
1- \Tr\big(P(1)\rho(1)\big) &\leq \frac{\norm{H'}}{\lambda g}\Big|_{s=0} + \frac{\norm{H'}}{\lambda g}\Big|_{s=1} + \int_0^1\bigg(\frac{\norm{H''}}{\lambda g} + 4\frac{\norm{H'}^2}{\lambda g^2} + \Big|\Big(\frac{1}{\lambda}\Big)'\Big|\frac{\norm{H'}}{g}\bigg)\diff{s}.
\end{align}
\end{theorem}
\begin{proof}
This is an immediate application of \autoref{mainTheorem}. The bounds on $\norm{P'}$ and $\norm{P''}$ are from \autoref{derivativeProjectorNormBound}, \autoref{projectorPrimesAndTildes} and \autoref{normBoundXTildeFiniteSetEigenvalues}.
\end{proof}

This implies that it is sufficient to take $\lambda = O\Big(\frac{\norm{H'}^2}{g^2} + \frac{\norm{H''}}{g}\Big)$. This is the same basic scaling in complexity as the previous methods, since, due to \autoref{phaseRandomisation}, the time complexity scales as $\lambda / g$.

\begin{theorem} \label{theorem:adaptiveRate}
Let $H(s)$ be a path of bounded Hamiltonians that is twice continuously differentiable. Suppose $\sigma(H(s))\cap [b_0(s), b_1(s)]$ consists of exactly one eigenvalue and \autoref{hyp:integralOrderReduction} holds.
Take $C$ such that
\begin{equation}
  C \geq \max_{s\in [0,1]}\Big(\big(2+(p-1)|g_0'|B_{3-p}\big)\norm{H'} + \norm{H''} + 4B_{3-p}\norm{H'}^2\Big).
\end{equation}
Fix $\epsilon > 0$ and set
\begin{equation}
  \lambda = \frac{1}{\epsilon} \frac{C}{g_0^{p-1}g_{0m}^{2-p}}
\end{equation}
Then \autoref{procedure} produces a state with fidelity $\Tr\big(P(1)\rho(1)\big) \geq 1-\epsilon$ and average time complexity bounded by
\begin{equation}
  \int_0^1 \frac{\lambda}{g_0} \diff{s} \leq \frac{t_0}{\epsilon}\frac{CB_{p}}{g_{0m}}. \label{eq:adaptedSchedulaPhaseRandomisationTimeComplexity}
\end{equation}
\end{theorem}
The constant $t_0$ can be taken to be $2.32132$.

\section{Applications} \label{sec:applications}
\subsection{Grover search}
For the Grover problem, we have an $N$-dimensional vector space we want to find an element of an $M$-dimensional subspace $\mathcal{M}$. In order to help us, we assume we have access to an oracle Hamiltonian $H_1 = \mathbb{1} - P_\mathcal{M}$, where $P_\mathcal{M}$ is the orthogonal projector on $\mathcal{M}$. In other words, we assume $H_1$ is admissible. We also assume $H_0 = \mathbb{1} - \ketbra{u}{u}$ is admissible, where $\ket{u} = \frac{1}{\sqrt{N}}\sum_{i=1}^N\ket{i}$ is the uniform superposition. The aim is now to use the interpolation $H(s) = (1-s)H_0 + sH_1$ to prepare a state in $\mathcal{M}$. For more details see \cite{farhiAQC} and \cite{Roland_2003}.

We see that $H(s)$ has four eigenvalues:
\begin{align}
\lambda_{1,2} &= \frac{1}{2}\left(1\pm \sqrt{1-4(1- \frac{M}{N})s(1-s)}\right) &\text{with multiplicity $1$} \\
\lambda_{3} &= 1-s &\text{with multiplicity $M-1$}\\
\lambda_{4} &= 1 &\text{with multiplicity $N-M-1$.}
\end{align}
The eigenvectors corresponding to $\lambda_3$ are the eigenvectors in $\mathcal{M}$ with zero overlap with $\ket{u}$. The eigenvectors corresponding to $\lambda_4$ are the eigenvectors in $\mathcal{M}^\perp$ with zero overlap with $\ket{u}$. Since the initial state has zero overlap with any of these vectors and they are eigenvectors of each $H(s)$, none of them are prepared by the procedure and everything happens in the two-dimensional space spanned by the eigenvectors associated to $\lambda_1$ and $\lambda_2$.

We have explicitly computed the gap, so we can use this as the bound $g$:
\begin{equation}
g(s) = \sqrt{1-4(1- \frac{M}{N})s(1-s)}. \label{eq:GroverGap}
\end{equation}
We can set $g_m = \min_{s\in [0,1]} g(s) = \sqrt{M/N}$. In order to give bounds on the time-complexity, we use the following result:
\begin{lemma} \label{lemma:GroverLemma}
For all $p > 1$ and $g$ given by \eqref{eq:GroverGap}, we have
\begin{equation}
\int_0^1 \frac{1}{g(s)^p}\diff{s} = O\big(\sqrt{N/M}^{p-1}\big) = O\big(g_m^{1-p}\big),
\end{equation}
and, for $p=1$,
\begin{equation}
\int_0^1 \frac{1}{g(s)}\diff{s} = O\big(\log(N/M)\big).
\end{equation}
\end{lemma}
We provide a proof in appendix \ref{appendix:GroverGap}. This implies that \autoref{hyp:integralOrderReduction} holds for all $1 < p < 2$.

We have $\norm{H'} = \norm{H_1 - H_0}$, $\norm{H''} = 0$ and
\begin{align}
|g'| &= \Big|\frac{4(1 - \frac{M}{N})(\frac{1}{2}-s)}{g}\Big| \\
&\leq \frac{2\sqrt{4(1 - \frac{M}{N})(\frac{1}{2}-s)^2}}{g} \\
&\leq \frac{2\sqrt{\frac{M}{N} + 4(1 - \frac{M}{N})(\frac{1}{2}-s)^2}}{g} = 2\frac{g}{g}  =2.
\end{align}
This implies that any of \autoref{theorem:adaptiveRateQubitisedDiscreteAdiabaticTheorem}, \autoref{theorem:adaptiveRateTimeIndependentDiscreteAdiabaticTheorem} or \autoref{theorem:adaptiveRate} achieve an optimal scaling of $O\big(\sqrt{N/M}\big)$.

\subsection{Solving linear systems of equations}
The Quantum Linear Systems Problem (QLSP) was introduced in \cite{HHL}. We make use of the Hamiltonian formulation from \cite{QLSPwithPR}. Suppose $A$ is an invertible $N\times N$ matrix $b \in \mathbb{C}^N$ a vector. The goal is to prepare the quantum state $\frac{A^{-1}\ket{b}}{\norm{A^{-1}\ket{b}}}$. We express the time complexity of our algorithm in terms of the condition number $\kappa = \norm{A}\,\norm{A^{-1}}$.

We may restrict ourselves to Hermitian matrices because we can use the following trick from \cite{HHL}: If $A$ is not Hermitian, we consider the matrix $\begin{pmatrix}0 & A \\ A ^* & 0 \end{pmatrix}$, which has the same condition number, and solve the equation $\begin{pmatrix}0 & A \\ A ^* & 0 \end{pmatrix}\ket{y} = \begin{pmatrix}\ket{b} \\ 0\end{pmatrix}$.

First we rescale the matrix $A$ to $\frac{A}{\norm{A}}$. We do this because typically admissible matrices need to be uniformly bounded. This has the effect of shifting the lowest singular value from $\frac{1}{\norm{A^{-1}}}$ to $\frac{1}{\norm{A}\norm{A^{-1}}} = \kappa^{-1}$. Now we consider a path of Hamiltonians that was introduced in \cite{QLSPwithPR}. Define $A(s) \defeq (1-s)\sigma_z\otimes \mathbb{1} + s\sigma_x\otimes A$,  $Q_{b,+} \defeq \mathbb{1} - \big(\ket{+}\ket{b}\big)\big(\bra{+}\bra{b}\big)$ and $\sigma_{\pm} \defeq \frac{1}{2} \big(\sigma_x \pm i \sigma_y\big)$. Set $H(s) = \sigma_+\otimes \big(A(s)Q_{b,+}\big) + \sigma_-\otimes \big(Q_{b,+}A(s)\big)$. This can be written as a linear interpolation $H(s) = (1-s)H_0 + sH_1$, where
\begin{align}
H_0 &\defeq \sigma_+\otimes \big((\sigma_z\otimes \mathbb{1})Q_{b,+}\big) + \sigma_-\otimes \big(Q_{b,+}(\sigma_z\otimes \mathbb{1})\big) \\
H_1 &\defeq \sigma_+\otimes \big((\sigma_x\otimes A)Q_{b,+}\big) + \sigma_-\otimes \big(Q_{b,+}(\sigma_x\otimes A)\big).
\end{align}
Following the analysis of \cite{QLSPwithPR}, we see that $H(s)$ has $0$ as an eigenvalue for all $s\in [0,1]$. The corresponding eigenspace is spanned by $\{\ket{0}\otimes \ket{x(s)}, \ket{1}\otimes \ket{+}\ket{b}\}$, where $\ket{x(s)} \defeq \frac{A(s)^{-1}\ket{b}}{\norm{A(s)^{-1}\ket{b}}}$. Since $H(s)$ does not allow transition between these states, we are sure to not prepare $\ket{1}\otimes \ket{+}\ket{b}$, so long as we start with $\ket{0}\otimes \ket{x(0)}$.

In \cite{QLSPwithPR} it was also shown that the eigenvalue zero is separated from the rest of the spectrum by a gap that is at least
\begin{equation} g(s) = \sqrt{(1-s)^2 + \Big(\frac{s}{\kappa}\Big)^2}. \label{LinSysGap} \end{equation}
If $\kappa$ is large enough, then we can take $g_m \defeq \frac{1}{2\kappa} \leq \sqrt{\frac{1}{\kappa^2 + 1}} = \min_{s\in [0,1]}g(s)$.

In order to give bounds on the time-complexity, we use the following result:
\begin{lemma} \label{lemma:QLSP}
For all $p > 1$, we have
\begin{equation}
\int_0^1 \frac{1}{g(s)^p}\diff{s} = O\big(\kappa^{p-1}\big) = O\big(g_m^{1-p}\big),
\end{equation}
and, for $p=1$,
\begin{equation}
\int_0^1 \frac{1}{g(s)}\diff{s} = O\big(\log(\kappa)\big).
\end{equation}
\end{lemma}
We provide a proof in appendix \ref{appendix:QLSPgap}. This implies that \autoref{hyp:integralOrderReduction} holds for all $1 < p < 2$.

We have $\norm{H'} = \norm{H_1 - H_0}$, $\norm{H''} = 0$ and
\begin{align}
|g'| &= \Big|\frac{s-1 + s/\kappa^2}{g}\Big| \\
&= \frac{\sqrt{(s-1 + s/\kappa^2)^2}}{g} \\
&= \frac{\sqrt{(1+1/\kappa^2)^2s^2 -(1+1/\kappa^2)2s + 1}}{g} \\
&\leq \frac{\sqrt{(1+1/\kappa^2)^2s^2 -(1+1/\kappa^2)2s + (1+1/\kappa^2)}}{g} \\
&= \sqrt{1+1/\kappa^2}\frac{g}{g} = \sqrt{1+1/\kappa^2} = O(1).
\end{align}
This implies that any of \autoref{theorem:adaptiveRateQubitisedDiscreteAdiabaticTheorem}, \autoref{theorem:adaptiveRateTimeIndependentDiscreteAdiabaticTheorem} or \autoref{theorem:adaptiveRate} achieve an optimal scaling of $O(\kappa)$.

For more details on implementing this algorithm using Poisson-distributed phase randomisation, as well as a more careful complexity analysis, see \cite{cunninghamTQC}.

The Poisson-distributed qubitised unitary algorithm is close to the deterministic algorithm of \cite{QLSPdiscreteAdiabaticTheorem} and we refer to this paper for more details.

\section{Conclusion}

In this work, we have proposed a general framework for deriving adiabatic theorems for quantum processes governed by linear differential equations. By exploiting the technique of Poissonisation, we were able to extend this analysis to discrete operations, including qubitized unitary evolution, phase randomisation, and Trotterised dynamics. This approach allows us to emulate the favourable features of adiabatic quantum computation without the necessity of simulating continuous time-dependent Hamiltonians. It also allows us to make use of Hamiltonians that may be difficult to implement physically, e.g.\ due to locality constraints.

We demonstrated the efficacy of this framework by applying it to the Grover search problem and the Quantum Linear Systems Problem (QLSP). In both cases, we showed that by adopting an adaptive schedule that respects the local spectral gap structure, it is possible to achieve optimal asymptotic scaling in the query complexity and condition number, respectively. 

Furthermore, we establish in our framework that the error bounds for Poissonised Trotter steps are asymptotically tighter than those derived from standard Trotter product formulas. This confirms that the cost of discretisation in the context of eigenpath traversal is often lower than worst-case general bounds would imply, reproducing key results of \cite{an2025largetimestepdiscretisationadiabatic} in a randomised setting. Collectively, these results suggest that Poissonised dynamics offer a versatile and efficient toolkit for designing and analysing quantum algorithms on digital quantum computers.

\vspace{1em}

This work was supported by the Belgian Fonds de la Recherche Scientifique - FNRS under Grants No. R.8015.21 (QOPT) and O.0013.22 (EoS CHEQS)

\printbibliography

\appendix

\section{The ``twiddle'' operation and bounds on derivatives of projectors}
\subsection{The ``twiddle'' operation} \label{sec:twiddleOperation}

In this section we consider a bounded normal operator $A$ on a Hilbert space $\mathcal{H}$. We assume there exists a simple closed curve $\Gamma$ that does not intersect the spectrum $\sigma(A)$ of $A$ and contains $m$ eigenvalues (which may be degenerate). Let $g$ be the distance between the part of $\sigma(A)$ inside $\Gamma$ and the part outside. This situation is illustrated in \autoref{fig:spectrumIllustrationA}.

We consider the spectral projector $P$ on the part of the spectrum contained in $\gamma$. It has the following form:
\begin{equation}
P = \frac{1}{2\pi i}\oint_\Gamma R_A(z) \diff{z}, \label{eq:RieszProjector}
\end{equation}
where $R_A(z) \defeq (z\mathbb{1} - A)^{-1}$ is the resolvent of $A$. This is known as the Riesz form of the projector. We also define $Q \defeq \mathbb{1} - P$.

For any bounded operator $X$, we consider the so-called ``twiddle operation''
\begin{equation}
\widetilde{X} \defeq \frac{1}{2\pi i}\oint_{\Gamma}R_{A}(z) X R_{A}(z) \diff{z}.
\end{equation}
This object is often used in the study of the adiabatic theorem, see \cite{Avron1987, Jansen2007}, but usually it is only defined for Hermitian $A$. Since we need to consider a more general situation, we give a fairly extensive treatment.

\begin{figure}[h!]
\centering
\begin{tikzpicture}[scale=0.8]
  \def\pointspec{1/1, 1.1/-.2, 0.5/2.5, -1.2/0.5, 0.3/-1.5, 2.2/1.8, -0.5/-0.5, 9/1.5, 8.5/2.2, 9/1, 9.1/-.2, 8.5/2.5, 7.2/0.5, 8.3/-1.5, 10.2/1.8, 11.5/-0.5, 10/1.5, 11.5/2.2, 5/1, 5/1.2, 4.9/1.4, 4.7/1.25}

  \filldraw[gray] plot [smooth cycle, tension=0.7] coordinates {
      (3.8, -2) (5.2, -2) (5.5, -2.3) (4.0, -2.8) (2.2, -2.1) (1.8, -1.5)
  };
  \filldraw[gray] plot [smooth cycle, tension=0.8] coordinates {
      (9.5, 2.0) (10.5, 2.8) (10.8, 4.0) (9.2, 3.5) (8.2, 2.5)
  };
  
  \draw[dashed, red] (4.9,1.2) circle (1) ++ (1,1) node {$\Gamma$};
  
  \draw[dashed, <->, shorten <= 5pt, shorten >= 5pt] (5,1) -- (7.2, .5) node[midway, above, xshift=2pt] {$g$};

  \draw[->, thick] (-1,0) -- (12,0);
  \draw[->, thick] (0,-3) -- (0,3);

  \foreach \x/\y in \pointspec{
    \draw[fill=black] (\x, \y) circle (.1);
  }
\end{tikzpicture}
\caption{An illustration of the spectrum of a normal operator $A$, with a non-intersecting path $\Gamma$ containing $m$ eigenvalues.} \label{fig:spectrumIllustrationA}
\end{figure}
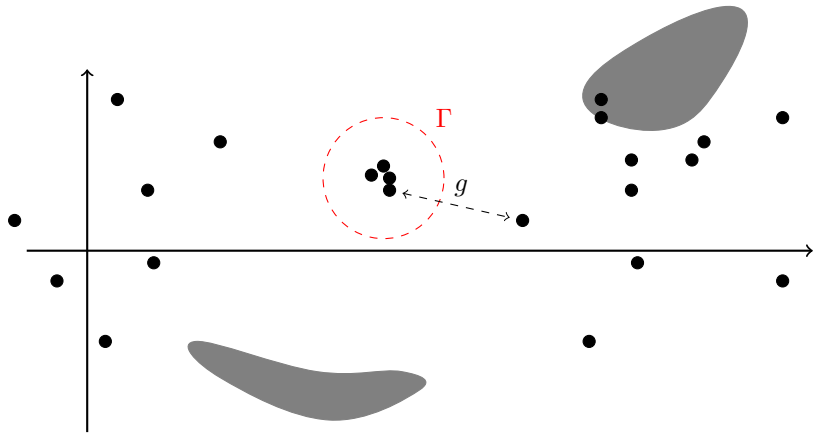

\begin{proposition} \label{spectralRegionSolution}
Let $A$ be a normal operator, $\Gamma$ a simple closed curve that does not intersect $\sigma(A)$ and $X$ a bounded operator. Then
\begin{enumerate}
\item $\comtor{P,X} = \comtor{A, \widetilde{X}}$;
\item $\widetilde{X} = P\widetilde{X}Q + Q\widetilde{X}P$.
\end{enumerate}
\end{proposition}
We can express property (2) by saying that $\widetilde{X}$ is ``off-diagonal''. 
\begin{proof}
(1) The proof is a straightforward verification of the proposed solution:
\begin{align}
\comtor{A,\widetilde{X}} &= \frac{1}{2\pi i}\oint_{\Gamma}\comtor[\big]{A, R_{A}(z) X R_{A}(z)} \diff{z} \\
&= \frac{1}{2\pi i}\oint_{\Gamma}\comtor[\big]{R_{A}(z) X R_{A}(z), z\id - A} \diff{z} \\
&= \frac{1}{2\pi i}\oint_{\Gamma}\big(R_{A}(z)X - X R_{A}(z)\big) \diff{z} \\
&= \Big(\frac{1}{2\pi i}\oint_{\Gamma}R_{A}(z)\diff{z}\Big)X - X\Big(\frac{1}{2\pi i}\oint_{\Gamma} R_{A}(z)\diff{z}\Big) \\
&= PX-XP = \comtor{P,X}.
\end{align}

(2) Since the spectrum $\sigma(A)$ is closed, the curve $\Gamma$ can be enlarged slightly without intersecting any of the spectrum. Call this new curve $\Gamma_1$. For any $z\in \Gamma$, consider the integral $\frac{1}{2\pi i}\oint_{\Gamma_1}\frac{R_A(w)}{w-z}\diff{w}$. In order to calculate this integral, we can deform $\Gamma_1$ such that it splits into two parts: one curve $\Gamma_2$ that lies inside $\Gamma$, but contains the same part of $\sigma(A)$ and another $\Gamma_3$ that is a small circle around $z$. See \autoref{fig:spectralRegionSolution}. Now
\begin{align}
\frac{1}{2\pi i}\oint_{\Gamma_1}\frac{R_A(w)}{w-z}\diff{w} &= \frac{1}{2\pi i}\oint_{\Gamma_3}\frac{R_A(w)}{w-z}\diff{w} + \frac{1}{2\pi i}\oint_{\Gamma_2}\frac{R_A(w)}{w-z}\diff{w} \\
&= R_A(z) - PR_A(z) = QR_A(z).\label{eq:QresolventIntegralFormula}
\end{align}
This allows us to calculate
\begin{align}
\widetilde{X}Q &= \frac{1}{2\pi i}\oint_{\Gamma}R_{A}(z) X R_{A}(z)Q \diff{z} \\
&= \frac{1}{(2\pi i)^2}\oint_{\Gamma}R_{A}(z) X \oint_{\Gamma_1}\frac{R_{A}(w)}{w-z}\diff{w}\diff{z} \\
&= \frac{1}{(2\pi i)^2}\oint_{\Gamma_1}\Big(\oint_{\Gamma}\frac{R_{A}(z)}{w-z} \diff{z}\Big) X R_{A}(w)\diff{w} \\
&= \frac{1}{(2\pi i)^2}\oint_{\Gamma_1}\Big(\oint_{\Gamma}\frac{R_{A}(z) - R_{A}(w)}{w-z} \diff{z}\Big) X R_{A}(w)\diff{w} \\
&= \frac{1}{(2\pi i)^2}\oint_{\Gamma_1}\Big(\oint_{\Gamma}R_A(w)R_A(z) \diff{z}\Big) X R_{A}(w)\diff{w} \\
&= \frac{1}{(2\pi i)^2}\Big(\oint_{\Gamma}R_A(z) \diff{z}\Big)\oint_{\Gamma_1}R_A(w) X R_{A}(w)\diff{w} \\
&= P\widetilde{X},
\end{align}
where the first resolvent identity has been used, as well as the fact that
\begin{equation}
\oint_{\Gamma} \frac{R_A(w)}{w - z}\diff{z} = 0,
\end{equation}
since the function is analytic inside $\Gamma$.
\end{proof}

\begin{figure}[h!]
\centering
\begin{tikzpicture}[scale=0.8]
  \def\r{0.8}        
  \def\start{0.7}   
  \def\theta{10}   
  \def\offsets{1.3, .5, .2, .7, .9}

  \def\pointspec{1/1,1.4/.8,1.7/-2, 1.75/1.5, 2.1/-1, 3.4/1, 3.6/1.2, 4/-1,7/.1, 8.7/.2, 9/.2, 9.1/.3, 11/2, 11.2/1.8, 11.8/-2,12/-2.2,12.1/-2,12.15/.2,14/.5,14.3/.12,14.6/.2,14.95/.2}
  \def\contspec{.7/.9, 2.2/3, 7.2/8.6, 13/13.5}
  
  \draw[->, thick] (-1,0) -- (16,0);
  \draw[->, thick] (0,-3) -- (0,3);

  \foreach \x\y in \pointspec{
  \draw[fill=black] (\x, \y) circle (.1);
  }

  \draw[red] (8,0) circle (2) ++(1.7,1.7) node {$\Gamma$};
  \draw[fill=black!30!green] (8,0) ++ (120:2) circle (.07) ++ (.2,.3) node {$z$};

  \draw[red] (8,0) circle (2.9) ++(2.3,2.3) node {$\Gamma_1$};

  \draw[blue] (8,0) ++ (120:2) ++ (-55:.6) arc[start angle=-55, end angle=300, radius=.6] -- ++(300:.55) arc[start angle=90, end angle=270, radius=.85] -- ++(1,0) arc[start angle=-90, end angle=90, radius=.85] -- ++(-.92,0) -- cycle;
  \draw[blue] (8,0) ++ (120:2) ++ (190:.9) node {$\Gamma_3$};
  \draw[blue] (8,0) ++ (0,-.93)  node[anchor=north] {$\Gamma_2$};
\end{tikzpicture}
\caption{An illustration of the integration contours used in the proof of \autoref{spectralRegionSolution}. The black dots represent the spectrum of $A$.} \label{fig:spectralRegionSolution}
\end{figure}
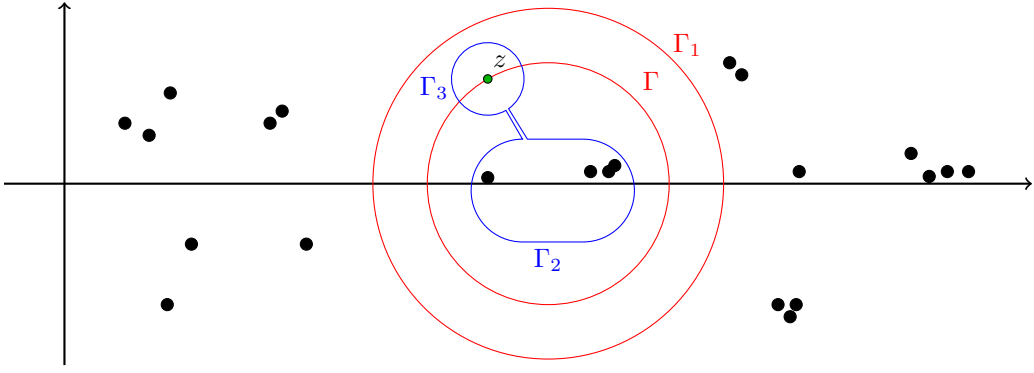

\begin{proposition} \label{uniquenessOffDiagonalSolution}
Let $A$ be a normal operator, $\Gamma$ a simple closed curve that does not intersect $\sigma(A)$ and $X$ a bounded operator. Then $\widetilde{X}$ is the unique off-diagonal solution $Y$ of $\comtor{P,X} = \comtor{A, Y}$.
\end{proposition}
\begin{proof}
Let $\lambda_{PA}$ be the operator on the space of bounded operators on $\mathcal{H}$ that implements left-multiplication by $PA$. Similarly let $\rho_{QA}$ be right-mutiplication operators.

Now, we have
\begin{equation}
(\lambda_{PA}-\rho_{QA})(Y) = PAY - YAQ = P(AY - YA)Q = P\comtor{A,Y}Q = P\comtor{P,X}Q = PXQ
\end{equation}
and
\begin{equation}
\sigma(\lambda_{PA}-\rho_{QA}) \subseteq \sigma(\lambda_{PA}) - \sigma(\lambda_{QA}) \subseteq \sigma(PA) - \sigma(QA) \subseteq \C\setminus\{0\},
\end{equation}
since $\lambda_{PA}$ and $\rho_{QA}$ commute and the spectra of $PA$ and $QA$ are disjoint. This implies that $\lambda_{PA}-\rho_{QA}$ is invertible, so $Y = (\lambda_{PA}-\rho_{QA})^{-1}(PXQ)$.
\end{proof}

\begin{lemma} \label{comparingOperatorEquationSolutions}
Let $A$ be a normal operator, $\Gamma$ a simple closed curve that does not intersect $\sigma(A)$ and $X$ a bounded operator. Suppose only a single eigenvalue $\omega_0$ is contained in $\Gamma$. Then
\begin{equation}
  \widetilde{X} = (\omega_0 \id - A)^+XP + PX(\omega_0 \id - A)^+,
\end{equation}
where the superscript ${}^+$ means taking the pseudoinverse\footnote{We consider the pseudoinverse $X^+$ as being defined as $f(X)$, where
\begin{equation}
  f: \C\to \C: x\mapsto \begin{cases}
  0 & (x=0) \\ x^{-1} & (\text{otherwise}).
  \end{cases}
\end{equation}
This requires $X$ to be normal, but is not restricted to finite dimensions. Note in particular that $X^+X = XX^+$ is the projector on the space orthogonal to the kernel of $X$.}
\end{lemma}
\begin{proof}
We calculate
\begin{align}
PXQ &= P\comtor{P,X}Q \\
&= P\comtor{A,\widetilde{X}}Q \\
&= \omega_0 P\widetilde{X}Q - P\widetilde{X}QA \\
&= P\widetilde{X}Q(\omega_0\id - A).
\end{align}
Multiplying both sides by $(\omega_0\id - A)^+$ on the right give $PX(\omega_0\id - A)^+ = P\widetilde{X}Q$. Similarly $Q\widetilde{X}P = (\omega_0\id - A)^+XP$. Since we also know that $\widetilde{X}$ is off-diagonal, we are done. 
\end{proof}

\begin{proposition} \label{comparingOperatorEquationSolutionsFullExpansion}
Let $A$ be a normal operator, $\Gamma$ a simple closed curve that does not intersect $\sigma(A)$ and $X$ a bounded operator. Suppose $\Gamma$ contains only the $m$ eigenvalues $\omega_0, \ldots, \omega_{m-1}$. Then
\begin{equation}
  \widetilde{X} = \sum_{k=0}^{m-1}(\omega_k \id - A)^+QXP_k + P_kXQ(\omega_k \id - A)^+,
\end{equation}
where $P_k$ projects onto the eigenspace associated by $\omega_k$.
\end{proposition}
\begin{proof}
Deform $\Gamma$ such that it breaks into $m$ separate curves, each circling one $\omega_k$. See \autoref{fig:comparingOperatorEquationSolutionsFullExpansion}. This deformation does not change the result, but does mean it can be written as the sum of $m$ simpler terms. Each one can be computed using \autoref{comparingOperatorEquationSolutions}. Finally, the factors of $Q$ can be added, since we know the operator is off-diagonal.
\end{proof}

\begin{figure}[h!]
\centering
\begin{tikzpicture}[scale=0.8]
  \def\r{0.8}        
  \def\start{0.7}   
  \def\theta{10}   
  \def\offsets{1.3, .5, .2, .7, .9}
  
  \draw[->, thick] (-1,0) -- (16,0);
  \draw[->, thick] (0,-3) -- (0,3);
  
  \draw[red] (\start,0) 
    arc[start angle=180, end angle=90, radius=\r]
    \foreach \dx in \offsets {
    arc[start angle=90, end angle=\theta, radius=\r] 
    -- ++(\dx, 0) 
    arc[start angle=180-\theta, end angle=90, radius=\r]}
    arc[start angle=90, end angle=0, radius=\r];
  \draw[red] (\start,0) 
    arc[start angle=180, end angle=270, radius=\r]
    \foreach \dx in \offsets {
    arc[start angle=270, end angle=360-\theta, radius=\r] 
    -- ++(\dx, 0) 
    arc[start angle=180+\theta, end angle=270, radius=\r]}
    arc[start angle=270, end angle=360, radius=\r];
  
  \draw[fill=black] (\start+\r, 0) circle (.1) node[below=4pt] {$\omega_0$}
  \foreach \dx in \offsets {
    ++ (2*\r+\dx-.02, 0) circle (.1) 
    } node[below=4pt] {$\omega_{m-1}$};
    
  \draw[red] (7,0) ellipse (7.85 and 2.5) ++(7,2) node {$\Gamma$};
  
  \draw[->, red, thick] (3, 2) -- (3.7, .9);
  \draw[->, red, thick] (7, 2.3) -- (7, 0.9);
  \draw[->, red, thick] (11, 2) -- (10.3, 0.9);
  \draw[->, red, thick] (3, -2) -- (3.7, -.9);
  \draw[->, red, thick] (7, -2.3) -- (7, -0.9);
  \draw[->, red, thick] (11, -2) -- (10.3, -0.9);
\end{tikzpicture}
\caption{An illustration of the proof of \autoref{comparingOperatorEquationSolutionsFullExpansion}.} \label{fig:comparingOperatorEquationSolutionsFullExpansion}
\end{figure}
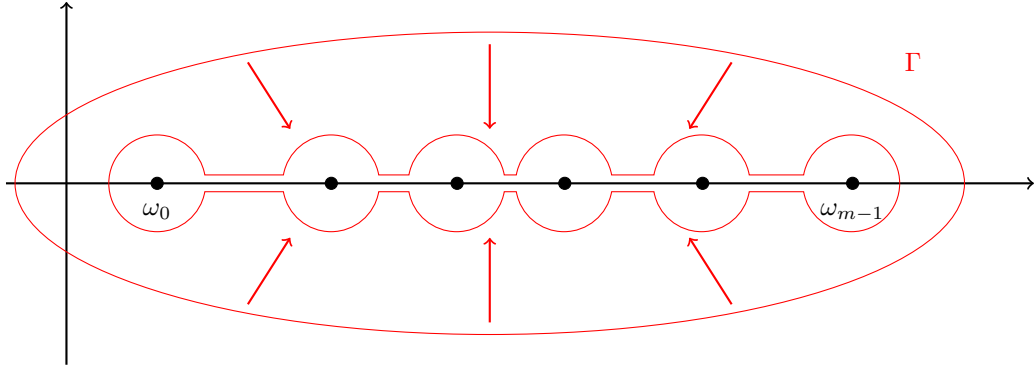

\subsection{Derivatives of twiddled operators}
As always, a prime denotes the derivative with respect to $s$.

\begin{proposition} \label{derivativeOfProjector}
Let $\mathcal{H}$ be a Hilbert space, $A(s)$ a bounded normal operator on $\mathcal{H}$ and $\Gamma(s)$ a closed simple curve that does not intersect the spectrum of $A(s)$, for all $s\in [0,1]$. Let $P(s)$ be the spectral projector on the spectrum of $A(s)$ inside $\Gamma(s)$. Suppose $A(s)$ is differentiable, then $P(s)$ is differentiable and $P' = \widetilde{A'}$.
\end{proposition}
\begin{proof}
By continuity, we can deform $\Gamma(s)$ to some constant curve $\Gamma_0$ on some neighbourhood of $s$. We use the Riesz form of the projector \eqref{eq:RieszProjector} to calculate
\begin{align}
\dod{P(s)}{s} &= \dod{}{s} \frac{1}{2\pi i}\oint_{\Gamma_0} R_A(z) \diff{z} \\
&= \frac{1}{2\pi i}\oint_{\Gamma_0} \dod{}{s} R_A(z) \diff{z} \\
&= \frac{1}{2\pi i}\oint_{\Gamma_0} R_A(z)A'R_A(z) \diff{z} = \widetilde{A'}.
\end{align}
\end{proof}

\begin{lemma}
Let $\mathcal{H}$ be a Hilbert space, $A$ a bounded normal operator and $\Gamma$ a closed simple curve in the complex plane that is disjoint from the spectrum $\sigma(A)$. Let $P$ the spectral projector on the part of the spectrum that lies inside $\Gamma$ and $X, Y$ bounded operators. Then
\begin{equation}
  \frac{1}{2\pi i}\oint_{\Gamma}R_{A}(z) X R_{A}(z)YR_{A}(z) \diff{z} = (Q-P)\big(\widetilde{X}\widetilde{Y} + \widetilde{X\widetilde{Y}} - \widetilde{\widetilde{X}Y}\big).
\end{equation}
\end{lemma}
\begin{proof}
First define $G(X,Y) \defeq \frac{1}{2\pi i}\oint_{\Gamma}R_{A}(z) X R_{A}(z)YR_{A}(z) \diff{z}$. Now note that
\begin{align}
\comtor{A,G(X,Y)} &= \frac{1}{2\pi i}\oint_{\Gamma}\comtor[\big]{H, R_{A}(z) X R_{A}(z)YR_{A}(z)} \diff{z} \\
&= -\frac{1}{2\pi i}\oint_{\Gamma}\comtor[\big]{z\id - A, R_{A}(z) X R_{A}(z)YR_{A}(z)} \diff{z} \\
&= \frac{1}{2\pi i}\oint_{\Gamma}\big(R_{A}(z) X R_{A}(z)Y - X R_{A}(z)YR_{A}(z)\big)\diff{z} \\
&= \widetilde{X}Y - X\widetilde{Y}.
\end{align}
Then
\begin{equation}
  \comtor{A,PG(X,Y)Q} = P\big(\widetilde{X}Y - X\widetilde{Y}\big)Q = \comtor[\big]{P, P\big(\widetilde{X}Y - X\widetilde{Y}\big)Q},
\end{equation}
so \autoref{uniquenessOffDiagonalSolution} gives
\begin{equation}
  PG(X,Y)Q = P\big(\widetilde{\widetilde{X}Y} - \widetilde{X\widetilde{Y}}\big)Q.
\end{equation}
Similarly, $\comtor{A,QG(X,Y)P} = -\comtor[\big]{P, Q\big(\widetilde{X}Y - X\widetilde{Y}\big)P}$, so
\begin{equation}
  QG(X,Y)P = -Q\big(\widetilde{\widetilde{X}Y} - \widetilde{X\widetilde{Y}}\big)P.
\end{equation}

With this we have derived the off-diagonal components of $G(X,Y)$. Now consider the diagonal components, starting with $QG(X,Y)Q$. The claim is that $QG(X,Y)Q = Q\widetilde{X}\widetilde{Y}Q$. This is verified by direct calculation. In the following, let $\Gamma_1$ be a curve that is slightly larger than $\Gamma$, but encircles the same spectral region, like in the proof of \autoref{spectralRegionSolution}.
\begin{align}
Q\widetilde{X}\widetilde{Y}Q &= \frac{1}{(2\pi i)^2} \Big(\oint_\Gamma QR_{A}(z)XR_{A}(z)\diff{z}\Big)\Big(\oint_{\Gamma_1} R_{A}(w)YQR_{A}(w)\diff{w}\Big) \\
&= \frac{1}{(2\pi i)^2}\oint_{\Gamma_1} \oint_{\Gamma} QR_{A}(z)XR_{A}(z)R_{A}(w)YQR_{A}(w)\diff{z}\diff{w} \\
&= \frac{1}{(2\pi i)^2}\oint_{\Gamma_1} \oint_{\Gamma} QR_{A}(z)X\Big(\frac{R_{A}(z) - \cancel{R_{A}(w)}}{w - z}\Big)YQR_{A}(w)\diff{z}\diff{w} \label{eq:tripleResolventCancellation1}\\
&= \frac{1}{(2\pi i)^2}\oint_{\Gamma} QR_{A}(z)XR_{A}(z)YQ\Big(\oint_{\Gamma_1}\frac{QR_{A}(w)}{{w - z}}\diff{w}\Big)\diff{z} \\
&= \frac{1}{2\pi i}\oint_{\Gamma} QR_{A}(z)XR_{A}(z)YQR_A(z)\diff{z} = QG(X,Y)Q.
\end{align}
The cancellation in \eqref{eq:tripleResolventCancellation1} still needs to be justified. This is due to the fact that
\begin{equation}
  \oint_{\Gamma} QR_{A}(z)X\frac{R_{A}(w)}{w - z}\diff{z} = 0,
\end{equation}
since the function in the integral is analytic (as a function of $z$) inside $\Gamma$: $w$ lies on $\Gamma_1$, which is outside $\Gamma$ and $QR_{H}(z)$ is only not analytic on the spectrum of $QH$, which also lies outside $\Gamma$. Also observe that $QR_{H}(w)$ is analytic inside $\Gamma_1$, so $\frac{1}{2\pi i}\oint_{\Gamma_1}\frac{QR_{A}(w)}{{w - z}}\diff{w} = QR_{A}(z)$ follows from Cauchy's integral formula.

Finally $PG(X,Y)P$ is calculated in a similar fashion.
\begin{align}
P\widetilde{X}\widetilde{Y}P &= \frac{1}{(2\pi i)^2} \Big(\oint_\Gamma PR_{A}(z)XR_{A}(z)\diff{z}\Big)\Big(\oint_{\Gamma_1} R_{A}(w)YPR_{A}(w)\diff{w}\Big) \\
&= \frac{1}{(2\pi i)^2}\oint_{\Gamma_1} \oint_{\Gamma} PR_{A}(z)XR_{A}(z)R_{A}(w)YPR_{A}(w)\diff{z}\diff{w} \\
&= \frac{1}{(2\pi i)^2}\oint_{\Gamma_1} \oint_{\Gamma} PR_{A}(z)X\Big(\frac{\cancel{R_{A}(z)} - R_{A}(w)}{w - z}\Big)YPR_{A}(w)\diff{z}\diff{w} \label{eq:tripleResolventCancellation2}\\
&= -\frac{1}{(2\pi i)^2} \oint_{\Gamma_1}\Big(\oint_{\Gamma}\frac{PR_{A}(z)}{{w - z}}\diff{z}\Big)XR_{A}(w)YQPR_{A}(w)\diff{w} \\
&= -\frac{1}{2\pi i}\oint_{\Gamma_1} PR_{A}(w)XR_{A}(w)YPR_A(w)\diff{z} = -PG(X,Y)P.
\end{align}
The cancellation in \eqref{eq:tripleResolventCancellation2} is due to
\begin{equation}
  \frac{1}{2\pi i}\oint_{\Gamma_1}\frac{PR_{A}(w)}{{w - z}}\diff{w} = PQR_A(z) = 0,
\end{equation}
from \eqref{eq:QresolventIntegralFormula}. Finally, we use the fact that all expressions under ~$\widetilde{}$~ are off-diagonal to collect terms and obtain the final result.
\end{proof}

\begin{proposition} \label{derivativeOfTwiddle}
Let $\mathcal{H}$ be a Hilbert space, $A(s)$ a bounded normal operator on $\mathcal{H}$, $X(s)$ a bounded operator on $\mathcal{H}$ and $\Gamma(s)$ a closed simple curve that does not intersect the spectrum of $A(s)$, for all $s\in [0,1]$. Suppose $A$ is differentiable and $\Gamma(s)$ is continuous. Let $P(s)$ be the spectral projector on the spectrum of $A(s)$ inside $\Gamma(s)$. Then
\begin{equation}
  (\widetilde{X})' = \widetilde{X'} + (Q-P)\big(P'\widetilde{X} + \widetilde{X}P' + \widetilde{\comtor{A',\widetilde{X}}} - \widetilde{\comtor{P',X}}\big).
\end{equation}
\end{proposition}
\begin{proof}
Using the fact that $\od{R_A(z)}{s} = R_A(z)A'R_A(z)$, we have
\begin{equation}
  (\widetilde{X})' = \widetilde{X'} + G(A', X) + G(X, A').
\end{equation}
The result follows with the observation that $\widetilde{A'} = P'$, \autoref{derivativeOfProjector}.
\end{proof}

\begin{lemma} \label{projectorPrimesAndTildes}
Let $\mathcal{H}$ be a Hilbert space, $A(s)$ a bounded normal operator on $\mathcal{H}$ and $\Gamma(s)$ a closed simple curve that does not intersect the spectrum of $A(s)$, for all $s\in [0,1]$. Let $P(s)$ be the spectral projector on the spectrum of $A(s)$ inside $\Gamma(s)$. Suppose $A(s)$ is differentiable, then
\begin{enumerate}
\item $\widetilde{P'}\comtor{P',P}+\widetilde{P'}' = \widetilde{P''} + (Q-P)\widetilde{\comtor{A', \widetilde{P'}}} + QP'\widetilde{P'}Q - PP'\widetilde{P'}P$;
\item $P'' = \widetilde{A''} + (Q-P)\big(2(P')^2 + \widetilde{\comtor{A', P'}}\big)$;
\item $\widetilde{P''} = \widetilde{\widetilde{A''}} + (Q-P)\widetilde{\widetilde{\comtor{A', P'}}}$.
\end{enumerate}
\end{lemma}
\begin{proof}
(1) \autoref{derivativeOfTwiddle} gives the expansion
\begin{align}
\widetilde{P'}' &= \widetilde{P''} + (Q-P)\big(P'\widetilde{P'} + \widetilde{P'}P' + \widetilde{\comtor{A', \widetilde{P'}}}\big) \\
&= \widetilde{P''} + (Q-P)\widetilde{\comtor{A', \widetilde{P'}}} + Q\big(P'\widetilde{P'} + \widetilde{P'}P'\big)Q - P\big(P'\widetilde{P'} + \widetilde{P'}P'\big)P,
\end{align}
which can be added to
\begin{equation}
  \widetilde{P'}\comtor{P',P} = P\widetilde{P'}P'P - Q\widetilde{P'}P'Q.
\end{equation}
to get the result.

(2) Since $P'' = \widetilde{H'}'$, the result follows from \autoref{derivativeOfTwiddle}.

(3) This holds because ~$\widetilde{}$~ kills the diagonal terms, is linear and commutes with multiplication by $P$ and $Q$.
\end{proof}

\subsection{Norm bounds}

\begin{proposition} \label{blockOperatorNormBound}
Let $\mathcal{H}$ be a Hilbert space, $P$ an orthogonal projector, $Q = \id - P$ and $X$ a bounded operator on $\mathcal{H}$. If
\begin{equation}
  X = PX_{0,0}P + PX_{0,1}Q + QX_{1,0}P + QX_{1,1}Q,
\end{equation}
then
\begin{equation}
  \norm{X} \leq \norm*{\begin{pmatrix}
\norm{X_{0,0}} & \norm{X_{0,1}} \\ \norm{X_{1,0}} & \norm{X_{1,1}}
\end{pmatrix}}.
\end{equation}
\end{proposition}
\begin{proof}
Take an arbitrary unit vector $\ket{\psi}\in \mathcal{H}$. Using the Pythagorean theorem gives
\begin{align}
\norm{X\ket{\psi}}^2 &= \norm{PX_{0,0}P\ket{\psi} + PX_{0,1}Q\ket{\psi}}^2 + \norm{QX_{1,0}P\ket{\psi} + QX_{1,1}Q\ket{\psi}}^2 \\
&\leq \big(\norm{X_{0,0}}\norm{P\ket{\psi}} + \norm{X_{0,1}}\norm{Q\ket{\psi}}\big)^2 + \big(\norm{X_{1,0}}\norm{P\ket{\psi}} + \norm{X_{1,1}}\norm{Q\ket{\psi}}\big)^2 \\
&= \norm*{\begin{pmatrix}
\norm{X_{0,0}}\norm{P\ket{\psi}} + \norm{X_{0,1}}\norm{Q\ket{\psi}} \\ \norm{X_{1,0}}\norm{P\ket{\psi}} + \norm{X_{1,1}}\norm{Q\ket{\psi}}
\end{pmatrix}}^2 \\
&= \norm*{\begin{pmatrix}
\norm{X_{0,0}} & \norm{X_{0,1}} \\ \norm{X_{1,0}} & \norm{X_{1,1}}
\end{pmatrix}\begin{pmatrix}
\norm{P\ket{\psi}} \\ \norm{Q\ket{\psi}}
\end{pmatrix}}^2 \leq \norm*{\begin{pmatrix}
\norm{X_{0,0}} & \norm{X_{0,1}} \\ \norm{X_{1,0}} & \norm{X_{1,1}}\end{pmatrix}}.
\end{align}
The final inequality is due to the fact that the column vector is a unit vector. Since this bound holds for all unit vectors $\ket{\psi}$, the norm bound holds.
\end{proof}
\begin{corollary} \mbox{} \label{normBoundPartitionedOperator}
\begin{enumerate}
\item If $X_{0,1} = 0 = X_{1,0}$, then
\begin{equation}
  \norm{X} \leq \max\{\norm{X_{0,0}}, \norm{X_{1,1}}\}.
\end{equation}
\item If $X_{0,0} = 0 = X_{1,1}$, then
\begin{equation}
  \norm{X} \leq \max\{\norm{X_{0,1}}, \norm{X_{0,1}}\}.
\end{equation}
\item If $X_{0,0} = 0 = X_{1,0}$, then
\begin{equation}
  \norm{X} \leq \sqrt{\norm{X_{0,1}}^2 + \norm{X_{1,1}}^2}.
\end{equation}
\item If $\norm{X_{0,1}} = \norm{X_{1,0}}$, then
\begin{equation}
  \norm{X} \leq \frac{\norm{X_{0,0}} + \norm{X_{1,1}}}{2} + \frac{1}{2}\sqrt{\big(\norm{X_{0,0}} - \norm{X_{1,1}}\big)^2 + 4\norm{X_{0,1}}^2}.
\end{equation}
\item If $X_{0,0} = 0$ and $\norm{X_{0,1}} = \norm{X_{1,0}}$, then
\begin{equation}
  \norm{X} \leq \frac{\norm{X_{1,1}}}{2} + \sqrt{\frac{\norm{X_{1,1}}^2}{4} + \norm{X_{0,1}}^2}.
\end{equation}
\end{enumerate}
\end{corollary}

\begin{proposition} \label{normBoundXTildeFiniteSetEigenvalues}
Let $\mathcal{H}$ be a Hilbert space, $A(s)$ a bounded normal operator on $\mathcal{H}$, $X(s)$ a bounded operator on $\mathcal{H}$ and $\Gamma(s)$ a closed simple curve that does not intersect the spectrum of $A(s)$, for all $s\in [0,1]$. Suppose $A$ is differentiable, $\Gamma(s)$ is continuous and the part of $\sigma(A)$ inside $\Gamma$ consists of $m$ eigenvalues. Then
\begin{equation}
  \norm{\widetilde{X}} \leq \sqrt{m}\frac{\max\{ \norm{PXQ}, \norm{QXP}\} }{g}.
\end{equation}
\end{proposition}
This bound is independent of the degeneracy of any of the eigenvalues, but it does depend on the number of distinct eigenvalues $m$.
\begin{proof}
From \autoref{blockOperatorNormBound}, we have $\norm{\widetilde{X}} \leq \max\{\norm{P\widetilde{X}Q}, \norm{Q\widetilde{X}P}\}$. Now $\norm{P\widetilde{X}Q}^2 = \norm{Q\widetilde{X}^*\widetilde{X}Q}$ and
\begin{align}
\norm{P\widetilde{X}Q}^2 = \norm{Q\widetilde{X}^*\widetilde{X}Q} &= \norm*{\sum_{k=0}^{m-1} (\overline{\omega_k}\id - A^*)^+QX^*P_kXQ(\omega_k\id - A)^+} \\
&\leq \sum_{k=0}^{m-1} \frac{\norm{P_kXQ}^2}{g^2} \leq m\frac{\norm{PXQ}^2}{g^2},
\end{align}
where we have used $\norm{(\omega_k\id - A)^+} = \frac{1}{g} = \norm{(\overline{\omega_k}\id - A^*)^+}$.
Similarly, $\norm{Q\widetilde{X}P}^2 = \norm{Q\widetilde{X}^*P\widetilde{X}Q} \leq m \frac{\norm{QXP}}{g^2}$ and we conclude with \autoref{comparingOperatorEquationSolutionsFullExpansion}.
\end{proof}
\begin{corollary} \label{derivativeProjectorNormBound}
Let $\mathcal{H}$ be a Hilbert space, $A(s)$ a bounded normal operator on $\mathcal{H}$ and $\Gamma(s)$ a closed simple curve that does not intersect the spectrum of $A(s)$, for all $s\in [0,1]$. Let $P(s)$ be the spectral projector on the spectrum of $A(s)$ inside $\Gamma(s)$ and suppose this consists of $m$ isolated eigenvalues. Suppose $A(s)$ is differentiable, then $P(s)$ is differentiable and
\begin{equation}
  \norm{P'} = \norm{\widetilde{A'}} \leq \sqrt{m} \frac{\norm{A'}}{g}.
\end{equation}
\end{corollary}
\begin{proof}
This is immediate, using \autoref{derivativeOfProjector}.
\end{proof}

\begin{lemma} \label{normsOfderivativesOfTilde}
Let $X$ be a bounded operator that is twice continuously differentiable in $s$. Then
\begin{enumerate}
\item $\norm{(\widetilde{X})'} \leq \frac{\sqrt{m}}{g}\norm{X'} + 6m \frac{\norm{A'}}{g^2}\norm{X}$;
\item $\norm{\widetilde{X}''} \leq 64 m\sqrt{m}\frac{\norm{A'}^2}{g^3} \norm{X} + 6 m \frac{\norm{A''}}{g^2}\norm{X} + 12m \frac{\norm{A'}}{g^2}\norm{X'} + \frac{\sqrt{m}}{g}\norm{X''}$.
\end{enumerate}
\end{lemma}
\begin{proof}
(1) Straightforward from \autoref{derivativeOfTwiddle} and \autoref{normBoundXTildeFiniteSetEigenvalues}.

(2) We have
\begin{align}
(\widetilde{X})'' &= \widetilde{X'}' - P'\Big(P'\widetilde{X} + \widetilde{X}P' + \widetilde{\comtor{H',\widetilde{X}}} - \widetilde{\comtor{P', X}}\Big) \\
&\hspace{7em}+ (Q-P)\Big(P'\widetilde{X} + \widetilde{X}P' + \widetilde{\comtor{A',\widetilde{X}}} - \widetilde{\comtor{P', X}}\Big)' \nonumber  \\
&= \widetilde{X'}' - P'\Big(P'\widetilde{X} + \widetilde{X}P' + \widetilde{\comtor{A',\widetilde{X}}} - \widetilde{\comtor{P', X}}\Big) \label{eq:tildeDoublePrime} \\
&\hspace{5em}+ (Q-P)\Big(P''\widetilde{X} + \widetilde{X}P'' + P'\widetilde{X}' + \widetilde{X}'P' + \widetilde{\comtor{A',\widetilde{X}}}' - \widetilde{\comtor{P', X}}'\Big) \nonumber
\end{align}
Using (1), the norm of the first term of \eqref{eq:tildeDoublePrime} can be bounded by
\begin{equation}
  \norm*{\widetilde{X'}'} \leq \frac{\sqrt{m}}{g}\norm{X''} + 6m \frac{\norm{A'}}{g^2}\norm{X'}.
\end{equation}
The norm of the second term can be bounded by
\begin{equation}
  \norm[\Big]{P'\Big(P'\widetilde{X} + \widetilde{X}P' + \widetilde{\comtor{A',\widetilde{X}}} - \widetilde{\comtor{P', X}}\Big)} \leq 6m\sqrt{m}\frac{\norm{A'}^2}{g^3}\norm{X}.
\end{equation}
Now (1) can be used to bound
\begin{align}
  \norm[\Big]{\widetilde{\comtor{A', \widetilde{X}}}'} &\leq \frac{\sqrt{m}}{g}\norm[\big]{\comtor{A',\widetilde{X}}'} + 6m\frac{\norm{A'}}{g^2}\norm[\big]{\comtor{A', \widetilde{X}}} \\
  &\leq 24 m\sqrt{m}\frac{\norm{A'}^2}{g^3} \norm{X} + 2 m \frac{\norm{A''}}{g^2}\norm{X} + 2 m \frac{\norm{A'}}{g^2}\norm{X'}. \label{eq:thirdTermPart1}
\end{align}
Using
\begin{equation}
  \norm{P''} \leq \frac{\sqrt{m}}{g}\norm{A''} + 4 m \frac{\norm{A'}^2}{g^2},
\end{equation}
we also have
\begin{align}
\norm[\Big]{\widetilde{\comtor{P', X}}'} &\leq 2\frac{\sqrt{m}}{g}\big(\norm{P''}\norm{X} + \norm{P'}\norm{X'}\big) + 6 m \frac{\norm{A'}}{g^2}\norm{P'}\norm{X} \\
&\leq 14 m\sqrt{m}\frac{\norm{A'}^2}{g^3} \norm{X} + 2m \frac{\norm{A''}}{g^2}\norm{X} + 2m \frac{\norm{A'}}{g^2}\norm{X'}. \label{eq:thirdTermPart2}
\end{align}
Also
\begin{equation}
  2\norm{P''}\norm{\widetilde{X}} \leq 8 m\sqrt{m}\frac{\norm{A'}^2}{g^3}\norm{X} + 2 m \frac{\norm{A''}}{g^2}\norm{X} \label{eq:thirdTermPart3}
\end{equation}
and
\begin{equation}
  2\norm{P'}\norm{\widetilde{X}'} \leq 12 m\sqrt{m}\frac{\norm{A'}^2}{g^3} \norm{X} + 2m \frac{\norm{A'}}{g^2}\norm{X'}. \label{eq:thirdTermPart4}
\end{equation}
Finally, the third term of \eqref{eq:tildeDoublePrime} can be bounded by combining the bounds \eqref{eq:thirdTermPart1}, \eqref{eq:thirdTermPart2}, \eqref{eq:thirdTermPart3} and \eqref{eq:thirdTermPart4}:
\begin{equation}
  58 m\sqrt{m}\frac{\norm{A'}^2}{g^3} \norm{X} + 6 m \frac{\norm{A''}}{g^2}\norm{X} + 6m \frac{\norm{A'}}{g^2}\norm{X'}.
\end{equation}
Putting all these bounds together gives the result.
\end{proof}

\begin{lemma} \label{boundingTimeAdiabaticX}
Let $X$ be a bounded operator that is twice continuously differentiable in $s$. Then
\begin{enumerate}
\item $\norm{[P,P']'} \leq \frac{\sqrt{m}\norm{A''}}{g} + 2\frac{m\norm{A'}^2}{g^2}$;
\item $\norm[\big]{\widetilde{\comtor{P,P'}}'} \leq \frac{m}{g^2}\norm{PA''Q} + 5\frac{m\sqrt{m}}{g^3}\norm{A'}^2$.
\end{enumerate}
\end{lemma}
\begin{proof}
(1) We calculate, using \autoref{projectorPrimesAndTildes},
\begin{align}
\comtor{P,P'}' &= \cancel{P'P'} + PP'' - P''P - \cancel{P'P'} \\
&= P\widetilde{A''} - \cancel{2P(P')^2} - P\widetilde{\comtor{A',P'}} - \widetilde{A''}P + \cancel{2(P')^2P} - \widetilde{\comtor{A',P'}}P \\
&= P\widetilde{A''}Q - Q\widetilde{A''}P  - \widetilde{\comtor{A',P'}}.
\end{align}

(2) This follows from \autoref{derivativeOfTwiddle}, with the observation that $\comtor[\big]{P',\comtor{P,P'}}$ is diagonal and thus the twiddle operation sends it to zero.
\end{proof}

\section{Derivatives of operator-valued functions}

The following expression is most often stated for matrices, see \cite{delmoral2018taylorexpansionsquareroot} and the references therein, but also holds for bounded operators on infinite-dimensional Hilbert spaces.
\begin{proposition} \label{derivativeSquareRoot}
Let $A(t)$ be a differentiable path of (strictly) positive definite operators on a Hilbert space. Then $\sqrt{A(t)}$ is differentiable and
\[ \dod{\sqrt{A(t)}}{t} = \int_0^\infty e^{-s\sqrt{A(t)}}\dod{A(t)}{t}e^{-s\sqrt{A(t)}}\diff{s}. \]
\end{proposition}

The following is a standard result. For the matrix case, see \cite{Bhatia1997}. The operator case can be proved similarly.
\begin{proposition} \label{derivativeExponential}
Let $A(t)$ be a differentiable path of operators on a Hilbert space. Then $e^{A(t)}$ is differentiable and
\[ \dod{e^{A(t)}}{t} = \int_0^1 e^{(1-s)A(t)}\dod{A(t)}{t}e^{sA(t)}\diff{s}. \]
\end{proposition}

\section{Gap properties}
\subsection{The gap in the Grover problem} \label{appendix:GroverGap}
For the Grover problem we have the following gap:
\begin{equation}
g(s) = \sqrt{1-4(1- \frac{M}{N})s(1-s)}. \label{eq:GroverGapAppendix}
\end{equation}
We can set $g_m = \min_{s\in [0,1]} g(s) = \sqrt{M/N}$. We provide a proof of lemma \ref{lemma:GroverLemma}.

\begin{lemma} 
For all $p > 1$ and $g$ given by \eqref{eq:GroverGapAppendix}, we have
\begin{equation}
\int_0^1 \frac{1}{g(s)^p}\diff{s} = O\big(\sqrt{N/M}^{p-1}\big) = O\big(g_m^{1-p}\big),
\end{equation}
and, for $p=1$,
\begin{equation}
\int_0^1 \frac{1}{g(s)}\diff{s} = O\big(\log(N/M)\big).
\end{equation}
\end{lemma}
\begin{proof}
We note that $g(s)$ is symmetric about $s= 1/2$. It is also strictly decreasing on $[0,1/2]$, going from $1$ to a minimum of $\sqrt{M/N}$.  So we can write
\begin{align}
\int_0^1 \frac{1}{g(s)^p}\diff{s} &= 2\int_0^{1/2} \frac{1}{g(s)^p}\diff{s} \\
&= 2\Big(\int_0^{1/2- \sqrt{M/N}} \frac{1}{g(s)^p}\diff{s} + \int_{1/2- \sqrt{M/N}}^{1/2} \frac{1}{g(s)^p}\diff{s} \Big).
\end{align}
Since $g$ has a minimum of $\sqrt{M/N}$, we can bound the second integral by
\[ \int_{1/2- \sqrt{M/N}}^{1/2} \frac{1}{g(s)^p}\diff{s} \leq \sqrt{\frac{M}{N}}\Big(\frac{1}{\min_{s\in[0,1]}g(s)}\Big)^p = \frac{\sqrt{M/N}}{\sqrt{M/N}^p} = \sqrt{N/M}^{p-1}. \]
For the first integral, we write
\begin{align}
\int_0^{1/2 - \sqrt{M/N}} \frac{1}{g(s)^p}\diff{s} &= \int_1^{g\big(1/2- \sqrt{M/N}\big)} \frac{1}{g^p}\dod{s}{g}\diff{g} \\
&= \int_{g\big(1/2- \sqrt{M/N}\big)}^1 \frac{1}{g^p}\Big(-\dod{s}{g}\Big)\diff{g}.
\end{align}
We can invert \eqref{eq:GroverGap} to obtain $s = \frac{1}{2} - \frac{1}{2}\sqrt{1-\frac{1-g^2}{1-N/M}}$.
Then we have
\begin{equation}
-\od{s}{g} = \frac{g}{2\sqrt{(1-M/N)(g^2 - M/N)}}.
\end{equation}
We now calculate
\[ g\Big(\frac{1}{2} - \sqrt{\frac{M}{N}}\Big) = \sqrt{\frac{M}{N}}\sqrt{5 - 4 \frac{M}{N}} \geq 2\sqrt{\frac{M}{N}}, \]
assuming $M/N \leq 1/4$. So
\begin{align}
\int_0^{1/2 - \sqrt{M/N}} \frac{1}{g^p}\diff{s} &\leq \int_{2\sqrt{\frac{M}{N}}}^1 \frac{1}{g^p}\Big(-\dod{s}{g}\Big)\diff{g} \\
&= \int_{2\sqrt{\frac{M}{N}}}^1 \frac{1}{g^p} \frac{g}{2\sqrt{(1-M/N)(g^2 - M/N)}}\diff{g} \\
&\leq \int_{2\sqrt{\frac{M}{N}}}^1 \frac{1}{g^p} \frac{g}{2\sqrt{(1-M/N)(g^2 - g^2/4)}}\diff{g} \\
&= \frac{1}{\sqrt{3(1-M/N)}} \int_{2\sqrt{\frac{M}{N}}}^1 \frac{1}{g^p} \diff{g}.
\end{align}
Now $\frac{1}{\sqrt{3(1-M/N)}}$ is $O(1)$ and $\int_{2\sqrt{\frac{M}{N}}}^1 \frac{1}{g^p} \diff{g} = \Big[\frac{1}{(p-1)g^{p-1}}\Big]_{2\sqrt{M/N}}^1$ is $O\big(\sqrt{N/M}^{p-1}\big)$, if $p>1$. If $p=1$, then it is $O\big(\log\sqrt{N/M}\big)$.
\end{proof}

\subsection{The gap in QLSP} \label{appendix:QLSPgap}
For the quantum linear system problem we have the following bound on the gap:
\begin{equation} g(s) = \sqrt{(1-s)^2 + \Big(\frac{s}{\kappa}\Big)^2}. \label{LinSysGapAppendix} \end{equation}
If $\kappa$ is large enough, then we can take $g_m \defeq \frac{1}{2\kappa} \leq \sqrt{\frac{1}{\kappa^2 + 1}} = \min_{s\in [0,1]}g(s)$. We provide a proof of lemma \ref{lemma:QLSP}.
\begin{lemma}
For all $p > 1$, we have
\begin{equation}
\int_0^1 \frac{1}{g(s)^p}\diff{s} = O\big(\kappa^{p-1}\big) = O\big(g_m^{1-p}\big),
\end{equation}
and, for $p=1$,
\begin{equation}
\int_0^1 \frac{1}{g(s)}\diff{s} = O\big(\log(\kappa)\big).
\end{equation}
\end{lemma}
\begin{proof}
We note that $g(s)$ is strictly decreasing on $\Big[0,1- \frac{1}{\kappa^2 + 1}\Big]$, going from $1$ to a minimum of $\sqrt{\frac{1}{\kappa^2 + 1}}$.  So we can write
\[ \int_0^1 \frac{1}{g(s)^p}\diff{s} = \int_0^{1- \frac{1}{\kappa^2 + 1}} \frac{1}{g(s)^p}\diff{s} + \int_{1- \frac{1}{\kappa^2 + 1}}^1 \frac{1}{g(s)^p}\diff{s}. \]
Since $g$ has a minimum of $\sqrt{\frac{1}{\kappa^2 + 1}}$, we can bound the second integral by
\[ \int_{1- \frac{1}{\kappa^2 + 1}}^{1} \frac{1}{g(s)^p}\diff{s} \leq \frac{1}{\kappa^2 + 1}\Big(\frac{1}{\min_{s\in[0,1]}g(s)}\Big)^p = \frac{1}{\kappa^2 + 1}\big(\kappa^2 + 1\big)^{p/2} = \big(\kappa^2 + 1\big)^{p/2-1}. \]
For the first integral, we write
\begin{align}
\int_0^{1 - \frac{1}{\kappa^2+1}} \frac{1}{g^p}\diff{s} &= \int_1^{g\big(1 - \frac{1}{\kappa^2+1}\big)} \frac{1}{g^p}\dod{s}{g}\diff{g} \\
&= \int_{g\big(1 - \frac{1}{\kappa^2+1}\big)}^1 \frac{1}{g^p}\Big(-\dod{s}{g}\Big)\diff{g} \\
&= \int_{\sqrt{\frac{1}{\kappa^2+1}}}^1 \frac{1}{g^p}\Big(-\dod{s}{g}\Big)\diff{g}.
\end{align}
We can invert \eqref{LinSysGapAppendix} on $\Big[0,1- \frac{1}{\kappa^2 + 1}\Big]$ to obtain $s = \frac{\kappa^2}{\kappa^2+1}(1-g)$.
Then we have
\begin{equation}
-\od{s}{g} = \frac{\kappa^2}{\kappa^2+1},
\end{equation}
so
\begin{align}
\int_0^{1 - \frac{1}{\kappa^2+1}} \frac{1}{g^p}\diff{s} &= \int_{\sqrt{\frac{1}{\kappa^2+1}}}^1 \frac{1}{g^p}\frac{\kappa^2}{\kappa^2+1}\diff{g} \\
&= \frac{\kappa^2}{\kappa^2+1}\Big(\frac{1}{(p-1)g^{p-1}}\Big)\Big|^{g = \sqrt{\frac{1}{\kappa^2+1}}}_{g = 1} \\
&= O(\kappa^{p-1}).
\end{align}
If $p=1$, then the integral is $O\big(\log(\kappa)\big)$.
\end{proof}

\end{document}